\documentclass{article}
\usepackage[margin=2.5cm]{geometry} \usepackage{graphicx} \usepackage[utf8]{inputenc}
\usepackage[english]{babel}	
\usepackage{csquotes}				
\usepackage{amsmath,amssymb,mathtools,amsthm}
\usepackage[bb=boondox]{mathalfa}
\usepackage{hyperref} 
\usepackage{xcolor}
\usepackage{subcaption}
\usepackage{pgfplots}
\usepackage{cleveref}
\pgfplotsset{compat=1.15}
\usepackage{mathrsfs}
\usepackage[export]{adjustbox}
\usepackage{afterpage}
\usepackage{longtable}
\usepackage{enumitem}
\usepackage[normalem]{ulem}

\usepackage{array,multirow,graphicx,tabularx}
\usepackage{tabu}
\usepackage{algorithm}
\usepackage{algpseudocode}
\usepackage[backend=biber]{biblatex}
\usetikzlibrary{arrows,spy}

\newcommand{\tablepath}{DUMMY}

\algnewcommand\algorithmicinput{\textbf{Initialization:}}
\algnewcommand\INPUT{\item[\algorithmicinput]}

\newtheorem{theorem}{Theorem}
\newtheorem{lemma}[theorem]{Lemma}
\newtheorem{definition}[theorem]{Definition}
\newtheorem{corollary}[theorem]{Corollary}
\newtheorem*{remark}{Remark}
\newtheorem{assumption}[theorem]{Assumption}

\newcommand{\distrib}[2]{
\left\langle #1,#2\right\rangle
}

\newcommand{\RR}{\mathbb{R}}

\newcommand{\scalarproduct}[2]{\left\langle #1,#2\right\rangle}
\renewcommand{\d}{\textup{d}}
\newcommand{\diff}{\textup{d}}
\newcommand{\norm}[1][\cdot]{\left\lVert #1\right\rVert}
\newcommand{\R}{\mathcal{R}}

\newcommand{\D}{\mathcal{D}}
\newcommand{\CC}{\mathbb{C}}
\newcommand{\ds}{\displaystyle}
\newcommand{\N}{\mathbb{N}}

\newcommand{\ee}{\textup{e}}

\newcommand{\I}{\ensuremath{\mathrm{i}}}

\renewcommand{\P}{\mathcal{P}}
\renewcommand{\S}{\mathcal{S}}
\newcommand{\T}{\mathcal{T}}
\newcommand{\A}{\mathcal{A}}

\DeclareMathOperator{\range}{range}

\DeclareMathOperator*{\argmin}{arg\,min}

\newcommand{\fvec}{\mathrm{f}}
\newcommand{\gvec}{\mathrm{g}}
\newcommand{\ai}[3]{{#1}_{#2}^{#3}}

\DeclareFontFamily{U}{wncy}{}
\DeclareFontShape{U}{wncy}{m}{n}{<->wncyr10}{}
\DeclareSymbolFont{mcy}{U}{wncy}{m}{n}
\DeclareMathSymbol{\Sh}{\mathord}{mcy}{"58} 

\newlength{\imWidth}
\newcommand\maga{3}
\newcommand\zshift{0.05}

\makeatletter 
\define@key{spyfig@keys}{spywidth}{\def\spyfig@spywidth{#1}}%
\define@key{spyfig@keys}{spyheight}{\def\spyfig@spyheight{#1}}%
\define@key{spyfig@keys}{spyposx}{\def\spyfig@spyposx{#1}}%
\define@key{spyfig@keys}{spyposy}{\def\spyfig@spyposy{#1}}%

\newcommand{\spyfig}[2][spywidth=0.16,spyheight=0.16,spyposx=-0.24,spyposy=0.33]{
    \setkeys{spyfig@keys}{#1}
    \begin{tikzpicture}[    
			baseline=(current bounding box.north),
			spy using outlines={rectangle,white,dashed,magnification=\maga,width=\maga*\spyfig@spywidth\imWidth, height=\maga*\spyfig@spyheight\imWidth, connect spies}]
			\node {
				\includegraphics[height=\imWidth, width=\imWidth,keepaspectratio]{#2}
			};
			\spy on (\spyfig@spyposx\imWidth,\spyfig@spyposy\imWidth) in node [left] at (0.5\imWidth+\zshift\imWidth,-0.5\imWidth-\zshift\imWidth);
		\end{tikzpicture}
}
\makeatother

\newcommand{\vx}{\mathrm{x}}
\newcommand{\vz}{\mathrm{z}}

\newcommand{\vy}{\mathrm{y}}

\graphicspath{{Simus}}

\newcommand{\bil}[2]{\left\langle #1,#2\right\rangle}
\renewcommand{\d}{\textup{d}}

\title{The method of the approximate inverse for limited-angle CT}
\author{Bernadette Hahn, Ga\"el Rigaud and Richard Schmähl}
\date{} 

\begin{document}

\maketitle

\begin{abstract}
Limited-angle computerized tomography stands for one of the most difficult challenges in imaging. Although it opens the way to faster data acquisition in industry and less dangerous scans in medicine, standard approaches, such as the filtered backprojection (FBP) algorithm or the widely used total-variation functional, often produce various artefacts that hinder the diagnosis. With the rise of deep learning, many modern techniques have proven themselves successful in removing such artefacts but at the cost of large datasets. 
In this paper, we propose a new model-driven approach based on the method of the approximate inverse, which could serve as new starting point for learning strategies in the future. In contrast to FBP-type approaches, our reconstruction step consists in evaluating linear functionals on the measured data using reconstruction kernels that are precomputed as solution of an auxiliary problem. With this problem being uniquely solvable, the derived limited-angle reconstruction kernel (LARK) is able to fully reconstruct the object without the well-known streak artefacts, even for large limited angles. However, it inherits severe ill-conditioning which leads to a different kind of artefacts arising from the singular functions of the limited-angle Radon transform. The problem becomes particularly challenging when working on semi-discrete (real or analytical) measurements. We develop a general regularization strategy, named constrained limited-angle reconstruction kernel (CLARK), by combining spectral filter, the method of the approximate inverse and custom edge-preserving denoising in order to stabilize the whole process. We further derive and interpret error estimates for the application on real, i.e. semi-discrete, data and we validate our approach on synthetic and real data. 
\\[1em] 
 
\textbf{Keywords:} Limited-angle CT, severe, ill-posedness, approximate inverse, spectral filtering, edge-preserving denoising, approximation with radial basis functions, error analysis.
\end{abstract}
 
\section{Introduction}
 
Tomography is a powerful imaging technique that allows for the visualization of an object's interior structure based on indirect measurements. In computerized tomography (CT), the data are acquired by rotating an X-ray source around the object while emitting X-ray beams through the specimen whose intensity loss is recorded at a detector panel. From a mathematical point of view, the measured data correspond to integrals along straight lines of a function $f$ describing the X-ray attenuation coefficient of the investigated object. In two dimensions, the operator that maps $f$ into the set of its line integrals is the Radon transform
\begin{align*} {\cal R} f(s,\theta) = \int_{\mathbb{R}^2} f(\vx) \delta(s-\vx^\top  \theta) \,\mathrm{d} \vx , \end{align*}
with $s \in \mathbb{R}$ and $\theta\in S^1 = \{(\cos\phi,\sin\phi), \ \phi \in [0,2\pi)\}$, where $S^1$ denotes the unit circle. Thus, the core task in CT (with parallel geometry) lies in recovering $f$ supported in the unit disk from its projection data 
$$
g(s,\theta) = \R f(s,\theta), \quad s\in(-1,1),\ \theta = (\cos\varphi,\sin\varphi)^\top \text{ with } |\varphi|\leq \pi/2.
$$ 
The solution of this inverse problem is well understood and a variety of reconstruction algorithms have been developed over the years, most notably the standard filtered backprojection algorithm (FBP) \cite{Natterer86}.
 
However, due to physical, mechanical or safety limitations, many industrial and medical applications (e.g. laminography, electron microscopy or dental tomography) come with the constraint that the projection data can only be measured for a restriced angluar range $C_\Phi :=\{(\cos\varphi,\sin\varphi)^\top :  |\varphi | \leq (\pi - \Phi)/2\}$ with $\Phi > 0$. This scenario is referred to as \emph{limited-angle CT}. Throughout the article, we denote by $\R_\Phi$ the respective forward operator which is given by restricting $\R$ to $C_\Phi$. Applying standard solution techniques such as FBP to limited-angle data typically leads to reconstructions with blurred regions, obscured structural details and (streak) artefacts, cf. Figure \ref{fig:Data_FBP_SL}. This highlights the need for reconstruction methods specifically tailored to the limited-angle setting. 

\begin{figure}[!h]
    \centering
    \begin{subfigure}{0.24\linewidth}\centering
    	\includegraphics[width=\linewidth]{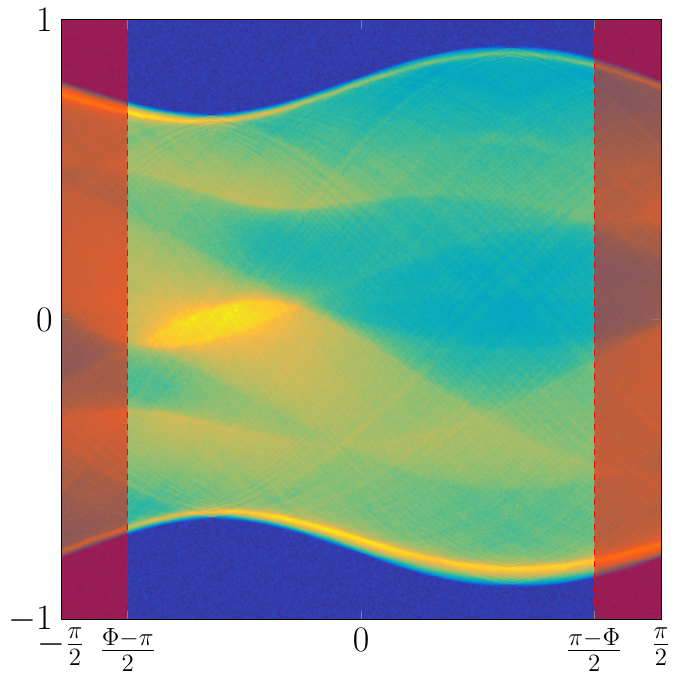}
        \caption{}
    \end{subfigure}
    \begin{subfigure}{0.24\linewidth}\centering
		\includegraphics[width=\linewidth]{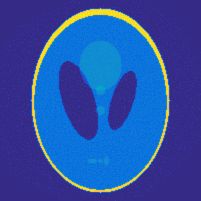}
        \caption{}
	\end{subfigure}	
	\begin{subfigure}{0.24\linewidth}\centering
		\includegraphics[width=\linewidth]{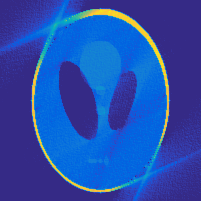}
        \caption{}
	\end{subfigure}	
    \caption{Illustration of standard limited-angle reconstructions for the Shepp-Logan phantom on a grid 201x201. For a better display in the diagonal of the limited-angle artefacts, the phantom is rotated of $45^\circ$ regarding the CT-scan. (a) Noisy data $g^\delta$ with in red the cut-off area here for $\Phi=40^\circ$. (b-c) FBP reconstruction using the Shepp-Logan-filter for $\Phi=0^\circ$ and $\Phi=40^\circ$, respectively.}
    \label{fig:Data_FBP_SL}
\end{figure}

In general, tomographic reconstruction is an ill-posed inverse problem, \textit{i.e.} small perturbations in the data, e.g. due to measurement errors, can lead to a large, unbounded error in the solution. Thus, the reconstruction step requires regularization methods in order to balance the influence of the data error and accuracy in the solution. While the full-data problem is only mildly ill-posed (more precisely of order 1/2 in the Sobolev scale \cite{Natterer86}), the inverse problem in the limited-angle case is severely ill-posed. This is revealed in the singular value decomposition derived by Alfred Louis in \cite{louis1986incomplete}. For a discussion of the mapping properties between Sobolev spaces, we refer to \cite{Natterer86}. This severe ill-posedness means that data errors are amplified exponentially, which makes the regularization of this type of problem particularly challenging.  \\%[1em]

To solve the limited-angle problem, several approaches have been investigated in the literature which need at the same time to handle noise and to somehow account for the missing data in order to reduce artefacts. 
 
A first class of methods focuses on data completion. In \cite{louis1980picture}, for instance, Alfred Louis developed a procedure based on the Helgason-Ludwig consistency conditions and a spectral decomposition to extrapolate the measured data to the missing angular range. Other methods rely solely on the moment conditions \cite{peres1979} or estimate the missing data in Fourier domain using analytic continuation \cite{inouye1979}.
 
A key to designing tailored reconstruction algorithms lies in the deep understanding of how the missing angular range affects the overall information content in the data -- specifically, which object features can be reliably reconstructed and what causes the undesirable artefacts. In \cite{louis1986incomplete}, this was studied in terms of the singular value spectrum. A stability analysis revealed that the components of the object connected with the singular functions belonging to large singular values can well be reconstructed. The smaller the missing range the larger is the number of these functions. Furthermore, the structure of the singular functions associated to small singular values can help to identify which features in the image are stable and which arise as artefacts due to the angular limitation.
 
Since the Radon transform is an elliptic Fourier integral operator, techniques from microlocal analysis can be used to classify which singularities can, and which cannot, be stably reconstructed from limited CT data \cite{quinto1993}. A characterization of the artefacts arising in FBP- and $\Lambda$-reconstruction was derived in \cite{frikel2013} which explains precisely where and why artefacts are created by these algorithms, followed by a suggestion on how to mitigate these artefacts. However, although the streak artefacts were suppressed, still structural information remain missing as they are not "{}seen"{} by the reconstruction operator. 
 
To address this fundamental limitation, variational regularization approaches gained a lot of interest since they allow the incorporation of suited priors, often guided by the theoretical analysis described above. 
Proposed priors include classic total variation (TV) \cite{persson2001total}, anisotropic \cite{chen2013limited} and directional TV \cite{zhang2021directional}, weighted relative structures \cite{gong2023structure}, as well as wavelets or curvelets \cite{frikel2013sparse}. The main challenge consists in choosing the regularizer to account for both the missing data and the noise at the same time. Thus, combinations of different regularizers are studied as well, for instance TV and curvelets \cite{goeppel2024l1TV}. The regularizer can also be modelled by edge-preserving diffusion along a certain direction, and edge-preserving smoothing along the corresponding perpendicular direction \cite{xu2019image}. The variational framework has further been extended to recover an inpainted sinogram simultaneously with the reconstruction \cite{tovey2019directional}.
 
Growing computational resources have enabled data-driven strategies, particularly deep learning, to complement traditional model-based regularization in addressing the challenges associated with the limited-angle problem. A wide range of approaches has been explored in this context, including simple post-processing for reducing artefacts produced by classical reconstruction schemes like FBP \cite{antholzer2019deep, goppel2025data}, data-driven data completion \cite{guo2025advancing}, as well as end-to-end training methods which learn the mapping from partial projections to the scanned object \cite{barutcu2021limited, Germer23helsinkichallenge}. 

The black-box nature of deep learning can be partially mitigated by hybrid strategies that integrate model-based and data-driven components. For instance, \cite{bubba2019} starts from a sharelet-based variational framework and delegates the task of inferring the invisible shearlet coefficients to a trained neural network. Invisible here refers to the analysis in \cite{quinto1993}. This is not the only instance where insights from microlocal analysis have guided the design of data-driven methods: In \cite{bubba2021}, a convolutional neural network is trained to learn an unknown part of a pseudodifferential operator, and hence effectively a reconstruction operator in limited-angle CT. Another example is the method proposed in \cite{andrade2022} which simultaneously extracts images and their wavefront sets via neural networks, thereby filling in missing data in a way that is consistent with the propagation of singularities. Nevertheless, despite their success, machine learning approaches often suffer from limited interpretability and face difficulties in generalizing to such ill-posed/ill-conditioned inverse problems.

\subsubsection*{Main contribution of the article}
 
A powerful and versatile regularization strategy is given by the method of the approximate inverse. This method was introduced in \cite{louismaass90,louis96} to solve ill-posed problems $\A f = g$ with a linear operator $\A:X\to Y$ between Hilbert spaces X and Y. The underlying idea is to recover linear functionals $f^\gamma(\cdot)=\langle f,e^\gamma_\cdot\rangle$ of $f$ with chosen mollifier $e^\gamma$ from the measured data. By solving the auxiliary problem $\A^* \ai{\psi}{\vx}{\gamma} = \ai{e}{\vx}{\gamma}$ for the reconstruction point $\vx$, one obtains 
$$ 
f^\gamma (\vx) = \scalarproduct{f}{\ai{e}{\vx }{\gamma}} = \scalarproduct{f }{\A^* \ai{\psi}{\vx}{\gamma} } = \scalarproduct{g}{\ai{\psi}{\vx}{\gamma} },$$
i.e. we can recover $f^\gamma$ by evaluating linear functionals on the data that are induced by the \emph{reconstruction kernel} $\psi^\gamma_\vx$. 
The method has been extensively studied with respect to its regularization properties \cite{louis1999unified}, has been further generalized to a variety of settings (including e.g. Banach spaces \cite{schuster2010Banach}) and has been extended to the stable extraction of features (e.g. contours) directly from measured data \cite{louis2011feature}. 

Considering linear functionals with mollifier $e^\gamma$ attenuates and controls the high frequencies in the solution. By solving the auxiliary problem instead of the original inverse problem, the kernels $\psi^\gamma_\vx$ can be precomputed independently of the measured data and, with the analytically known mollifier, in particular independently of any measurement errors. Furthermore, this elegant approach allows to exploit invariance properties of the operator $\A$ to efficiently compute the reconstruction kernels.
 
From \cite{Natterer86}, we know that $f$ is uniquely determined by perfect data $g=\R_\Phi f$. The property tends to extend to the discrete case in which the solution can be numerically built artefact-free from the limited-angle measurement by standard optimization techniques. The issue is that this approach does not resist to any measurement errors. Truncating the smallest singular values does not help as the solution tends to converge towards the FBP reconstruction in that case. 

Following this logic, the auxiliary problem above is expected to be numerically solvable resulting in a reconstruction kernel able to recover the missing regions. In fact, a first step was taken by Alfred Louis himself in \cite{louis2006development}, where he derived a series representation of the reconstruction kernel in the continuous case based on the singular value decomposition of $\R_\Phi$. Its application, however, remains challenging due to the severe ill-posedness of the problem: Errors in the data still lead to severe artefacts, although their nature changes for the method of the approximate inverse, as we will see in \Cref{sec:semi_discrete-AI}. 

To overcome this challenge and to successfully apply the method to real measured data, we proceed step-by-step:

\begin{enumerate}
    \item[1.)] We first derive a \emph{limited-angle reconstruction kernel} (LARK) in the continuous case which accounts for the specific way $\R_\Phi$ encodes the frequencies of $f$ in data (\Cref{sec:continuousLARK}).     
    \item[2.)] Then, we combine the method of the approximate inverse with a constraint on the solution space to handle noisy or \textit{out-of-range} data in \Cref{sec:denoising} -- an approach we name \textit{constrained LARK} (CLARK).
    \item[3.)] In \Cref{sec:discrete}, we compute the \textit{simple} fully-discrete approach for  (C)LARK which delivers satisfactory results on noisy data but fails on semi-discrete data.
    \item[4.)] In order to extend this strategy to semi-discrete data (such as real measurements), we need to represent the unknown function $f$ by a suited interpolation operator. This operator is constructed and error estimates are derived in \Cref{sec:semi_discrete}. As the interpolation operator induces an increase in ill-conditioning, this necessitates a smoothing step of the data. Numerical results on analytical and real data demonstrate that LARK and CLARK deliver convincing results.  
\end{enumerate}

The article concludes with a summary of the main results and a brief outlook.
 
\section{The continuous limited-angle Radon transform} \label{sec:continuous}
 
\subsection{Mathematical background}\label{Sec:MathBackground}

We start by reviewing properties of the limited-angle forward operator $\R_\Phi$  from the literature that will play a central role throughout the article. As presented in the introduction, we denote the considered angular range by $C_\Phi :=\{(\cos\varphi,\sin\varphi)^\top :  |\varphi | \leq (\pi - \Phi)/2\}$ with $\Phi\in(0,\pi)$.

\begin{theorem}\label{theo:uniqueness} Let $C \subseteq S^1$ be a set of directions such that no non-trivial homogeneous polynomial vanishes on $C$. If $f\in C^\infty(\mathbb{R}^2)$ is compactly supported and $\R_\Phi f(\cdot,\theta) = 0$ for $\theta \in C$, then $f=0$. 
\end{theorem}
\begin{proof}
See \cite{Natterer86}.
\end{proof}
As a consequence of the analyticity of compact functions, this result establishes that compactly supported functions are uniquely determined by $\R_\Phi f$. In \cite{louis1986incomplete}, this conclusion was drawn from the singular value decomposition of $\R_\Phi$. In order to analyze the singular values and singular functions, we will restrict ourselves to the unit disk  $\Omega:=\{x\in \mathbb{R}^2, |x|\leq 1\}.$
\begin{theorem}\label{theo:SVD_R_Phi} The singular values of the operator
$$
\R_\Phi : L_2(\Omega) \to  L_2(Z_\Phi;w^{-1}) 
$$
with $w := \sqrt{1-s^2}$ and $Z_\Phi := [-1,1]\times C_\Phi$ are given by 
$$
\sigma_{ml} = 2\left( \frac{\pi}{m+1} \lambda_l\left( m+1,\frac{\Phi}{\pi} \right) \right)^{1/2}, \quad m \in \mathbb{N}_0, \ l = 0,1,\ldots,m,
$$
where $\lambda_l(m+1,\Phi/\pi)$ are the eigenvalues of the matrix
$$ 	A_m(\Phi):=
  \begin{pmatrix}
  a_0 & a_1 & \cdots & a_m \\
  a_1 & a_0 & \cdots & a_{m-1} \\
  \vdots & & & \\
  a_m & a_{m-1} & \cdots & a_0
		\end{pmatrix}, \qquad 
  a_l = \frac{1}{\pi}
  \begin{cases}
      2\Phi, \quad & l=0 \\
      \frac{\sin 2 l \Phi}{l}, & l\neq 0.
  \end{cases}
$$  
We denote by $(\sigma_{ml}, u_{ml}, v_{ml})_{m,l}$ the singular system of $\R_\Phi$, \textit{i.e.}
$$
\R_\Phi v_{ml} = \sigma_{ml} u_{ml}, \quad 
\R_\Phi^* u_{ml} = \sigma_{ml} v_{ml}, \quad m \in \mathbb{N}_0, \ l = 0,1,\ldots,m.
$$
Let $\lambda_l \coloneqq\lambda_l (m,\Phi)\in\RR$, $d_l \coloneqq d_l (m,\Phi)\in\RR^{m+1} (0\leq l \leq m)$ s. t. $1>\lambda_0\geq \lambda_2\geq \cdots\geq \lambda_{m}>0$ and $A(m,\Phi)d_l =\lambda_l$ with $\norm[d_l ]=1$. Then, it holds for the singular functions
\begin{equation*}
\begin{dcases}
    {v}_{ml  }(\vx)&:= \epsilon_l  \sum_{l=0}^{m}d_l  Q_{m\left|2l-m\right|}(\vx)\ee^{\I(2l-m)\arg(\vx)}\\
    {u}_{ml  }(s,\theta) &:= \frac{\epsilon_l}{\sqrt{\lambda_l }}w(s)U_m(s)  \sum_{l=0}^{m}d_l \ee^{\I(m-2l)\arg(\theta)}
\end{dcases}
\quad\text{with}\quad 
\epsilon_l \coloneqq 
\begin{cases}
1&l \text{ even}\\\I&l \text{ odd}
\end{cases},
\end{equation*}
where $U_m$ denotes the Chebyshev polynomials of second kind and $Q_{ml}(x)\coloneqq \norm[x]_2^l P_{\frac{m-l}{2}}^{(0,l)}\left(2\norm[x]_2^2-1\right)$ with $P_n^{(0,l)}$ the Jacobi Polynomials.
\end{theorem}
\begin{proof} See \cite{louis1986incomplete}. 
\end{proof}

{\color{black}
\begin{remark}
    Please note that the singular values and functions of $\R_\Phi$ also depend on the missing angular range denoted by $\Phi$. For the sake of readibility, the respective index is omitted in the singular system.
\end{remark}

    The singular functions are, in fact, real-valued, due to the symmetric Toeplitz structure of $A_m(\Phi)$ enforcing certain "symmetries" for the coefficients \cite{sleepian78,cantoni76}. Furthermore, the singular values decay exponentially as $m\to\infty$, revealing the severe ill-posedness of the limited angle problem. 

    \begin{remark}
    Working with the continuous system in practice remains quite challenging from a numerical point of view. 
     First, small errors from numerical computations break the mentioned "{}symmetries"{}, i.e. small imaginary parts remain, which in turn can cause inaccuracies in the real components. Furthermore, the fast decay of the eigenvalues of $A_m(\Phi)$, see \cite{Natterer86}, induces numerical instabilities which may alter not only the singular values but also their arrangement in descending order, leading to  difficulties in the multiplication with $\epsilon_l$. In addition, depending on the grid given by the detector, the system may also loose orthogonality. While the latter can be avoided using grids optimized for Chebyshev and trigonometric polynomials  (cf. \cite{etna_vol39_pp102-112}), the usual detector setups are build without such orthogonality problems in mind. 

Therefore, for numerical implementations, we make use of the singular system of the discretized forward operator, cf. Section \ref{sec:semi_discrete-AI}, rather than discretizing the system of the continuous case. 
    \end{remark}
 
}

In Figure \ref{fig:singular_and_kernel}, we illustrate the decay of the singular values as well as singular functions $v_{ml}$ to two singular values (one large and one small) for $\R_\Phi$ with $\Phi=30^\circ$ simulated on a $121\times 121$ grid. In Figure \ref{fig:singular_and_kernel} a), we clearly observe the exponential decay of the singular values. Comparing Figure \ref{fig:singular_and_kernel} c) and e) shows that the singular functions associated to the small singular values carry the missing directions, completing the missing cone in the Fourier domain which can be seen as analytic continuation. This is why standard regularizations such as truncated SVD, tend to deliver FBP-type reconstructions, since they cut off this information. Inspecting the singular functions in the frequency domain further shows that $v_{ml}$ associated with the smallest singular values carry information from high to low frequency ranges. Thus, $\R_\Phi$ scatters the high-frequency components of $f$ through the whole spectrum. This contrasts with the full data case, where $\R_\pi$ encodes the high-frequency components of $f$ through the singular functions associated with small singular values.

This observation will play an important role when we build our regularization. 
 
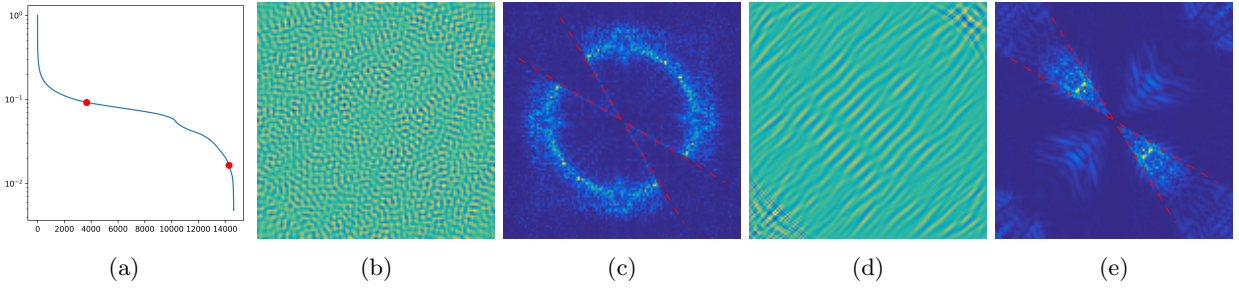
\begin{figure}[!h]
    \centering
    \begin{subfigure}{0.2\linewidth}
        \includegraphics[width=\linewidth]{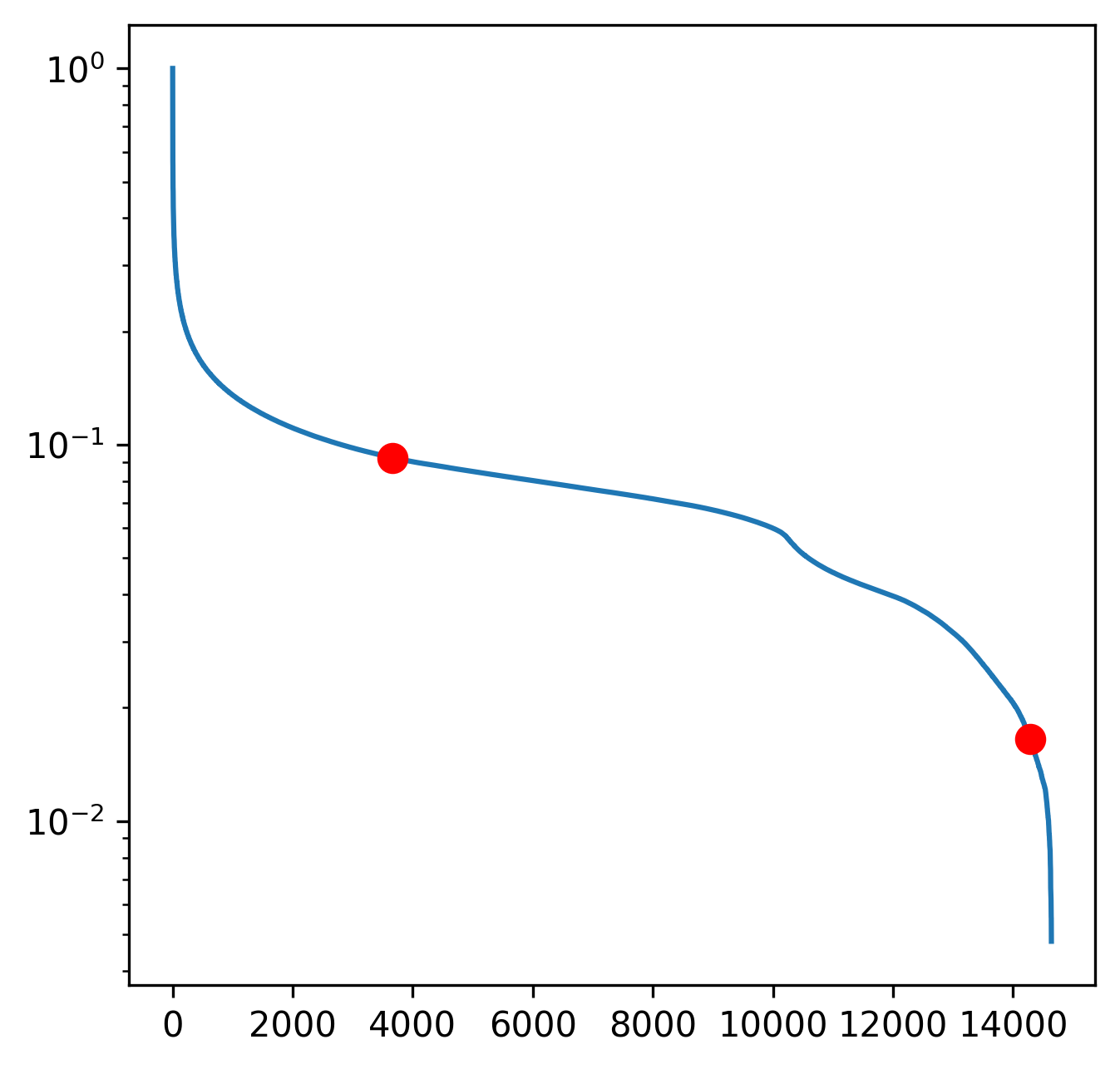}
        \caption{}
    \end{subfigure}
    \begin{subfigure}{0.19\linewidth}
        \includegraphics[width=\linewidth]{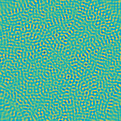}
        \caption{}
    \end{subfigure}
    \begin{subfigure}{0.19\linewidth}
    \begin{tikzpicture}[scale=0.55]    \node[inner sep=0pt] at(0,0) {\includegraphics[width=\linewidth]{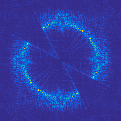}};
    \draw[dashed,red] (-2.5,1.5) -- (2.5,-1.5);
    \draw[dashed,red] (-1.5,2.5) -- (1.5,-2.5);
    \end{tikzpicture}
        \caption{}
    \end{subfigure}
    \begin{subfigure}{0.19\linewidth}
        \includegraphics[width=\linewidth]{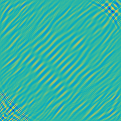}
        \caption{}
    \end{subfigure}
    \begin{subfigure}{0.19\linewidth}
    \begin{tikzpicture}[scale=0.55]    \node[inner sep=0pt] at(0,0) {\includegraphics[width=\linewidth]{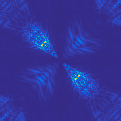}};
    \draw[dashed,red] (-2.5,1.5) -- (2.5,-1.5);
    \draw[dashed,red] (-1.5,2.5) -- (1.5,-2.5);
    \end{tikzpicture}
        \caption{}
    \end{subfigure}
    
    \caption{ Simulated on a 121x121 grid for $\Phi=30^\circ$. (a) Singular values of $\R_\Phi$ on a log-scale with two target singular values depicted as a red disk. (b-c) and resp. (d-e) Singular function $v_{ml}(\vx)$ and the magnitude of its 2D-Fourier transform for the first target singular value and the second one respectively. (c) and (e) highlight in red (dashed line) the missing cone.}     \label{fig:singular_and_kernel}
\end{figure}
 
\subsection{The method of the approximate inverse for $\R_\phi$} \label{sec:continuousLARK}
 
Our aim is to build a regularization for the limited-angle forward operator $\R_\Phi:L_2(\Omega)\to L_2(Z_\Phi;w^{-1})$, more precisely, to construct a family of continuous operators that approximates the generalized inverse, noted $\R_\Phi^\dag$, pointwise on $\mathrm{Ran}(\R_\Phi) \oplus \mathrm{Ran}(\R_\Phi)^\perp =: \mathcal{D}(\R_\Phi^\dagger)$, with $\mathrm{Ran}$ denoting the range of an operator.
 
The method of the approximate inverse \cite{louis96} proposes to build such a family by solving auxiliary problems $\R_\Phi^* \ai{\psi}{\vx}{\gamma} = \ai{e}{\vx}{\gamma}$ with a prescribed mollifier $\ai{e}{}{\gamma}$ and by evaluating linear functionals 
\begin{equation}\label{eq:def_S_gamma}
    S_\gamma g (\vx):= \langle g, \psi^\gamma_\vx\rangle.
\end{equation}
$S_\gamma$ is called the \emph{approximate inverse} of $\R_\Phi$ and $\psi^\gamma_\vx$ is called \emph{reconstruction kernel}. For the full-angle case, i.e. $\Phi=0$,  a closed-form representation of the reconstruction kernel $\ai{\psi}{\vx}{\gamma}$ can be obtained via the inversion formula of the Radon transform, see \cite{louis2008optimal,Schusterbook}. For $\Phi>0$, this is no longer feasible. However, the reconstruction kernel can at least be constructed via the singular value decomposition of $\R_\Phi$, see \cite{louis2006development}. \\

Before we address the representation of $\ai{\psi}{\vx}{\gamma}$, we provide a precise definition of a mollifier, following \cite{Schusterbook}.
 
\begin{definition}\label{def:mollifier}
    For all $\vx\in \Omega$ and for all $\gamma>0$, let $\ai{e}{\vx}{\gamma} \in L_2(\Omega)$ with 
    $$ \int_\Omega \ai{e}{\vx}{\gamma}(\vy)\,\mathrm{d}\vy = 1.$$
    Further let $$f^\gamma (\vx) := \int_\Omega f(\vy) \, \ai{e}{\vx}{\gamma}(\vy)\,\mathrm{d}\vy, \quad \vx \in \Omega$$
    converge to $f$ in $L_2(\Omega)$ as $\gamma \to 0$. Then, we call $\ai{e}{}{\gamma}$ a mollifier. 
\end{definition}

According to this definition, mollifers can be seen as (smooth) approximations to the delta distribution. 
For $g\in \mathrm{Ran}(\R_\Phi)$, i.e. $g=\R_\Phi f$ with $f\in L_2(\Omega)$, it holds
$$ S_\gamma g = \scalarproduct{\R_\Phi f}{\ai{\psi}{\vx}{\gamma}} = \scalarproduct{f}{\ai{e}{\vx}{\gamma}},$$
i.e. $S_\gamma g$ provides a smoothed version of $f$. 
Thus, by choosing the mollifer $e^\gamma_\cdot$, we guide the smoothness of our reconstruction. 
 
In order to regularize the inverse problem associated to $\R_\Phi$, we would like to start from a standard mollifier of convolution-type such as the widely-used Gaussian or a bump function. The issue here is that such a mollifier might not be in the range of $\R_\Phi^\star$ and therefore the auxiliary problem might not have a classic solution nor a minimum-norm solution. Nevertheless, one can approximate $\ai{e}{\vx}{\gamma}$ by a sequence of functions $(\ai{e}{\vx}{\gamma,n})_{n\in\N} \subset  \mathrm{Ran}(\R^*_\Phi)$.
\begin{lemma}\label{lemma:continuous}
Let $\ai{e}{\vx}{\gamma}$ be a $C^\infty$-smooth mollifier with compact support in $\Omega$. Then, there exists a sequence $(e^{\gamma,n}_{\vx})_{n\in\mathbb{N}} \subset \mathrm{Ran}(\R^*_\Phi)$ such that $e_{\vx}^{\gamma,n} \to e^{\gamma}_\vx$ for $n \to \infty$.
The reconstruction kernel $\psi_{\vx}^{\gamma,n}$ defined as the solution of $\R_\Phi^* \psi = e_{\vx}^{\gamma,n}$ for $n\in \mathbb{N}$ is then given by 
\begin{align}\label{eq:kernel_gamma_n}
\psi_{\vx}^{\gamma,n}= \sum_{m=0}^\infty\sum_{l=0}^{m} \frac{\scalarproduct{e_{\vx}^{\gamma,n}}{{v}_{ml}}}{{\sigma}_{ml}} u_{ml}.
\end{align}    
\end{lemma}
\begin{proof}
Theorem \ref{theo:uniqueness} ensures that $\mathrm{ker}(\R_\Phi) = \{0\}$ with $\R_\Phi$ defined on $L_2(\Omega)$. This implies that $\overline{\mathrm{Ran}(R_\Phi^*)}=\mathrm{ker}(\R_\Phi)^\perp = L_2(\Omega)$. The existence of such a sequence $(e_{\vx}^{\gamma,n})_{n\in\mathbb{N}}$ follows from the definition of the closure.\\
Since $e_{\vx}^{\gamma,n} \in \mathrm{Ran}(\R^*_\Phi)$, $n\in\mathbb{N}$, it satisfies the Picard criterion which means that the series of $\sigma_{ml}^{-1} \scalarproduct{e_{\vx}^{\gamma,n}}{{v}_{ml}}$ converges. The representation of $\psi_{\vx}^{\gamma,n}$ follows now from the singular value decomposition of the generalized inverse 
$$
(\R_\Phi^*)^\dag h:= \sum_{m=0}^\infty\sum_{l=0}^{m} \frac{\scalarproduct{h}{{v}_{ml}}}{{\sigma}_{ml}} u_{ml}.
$$
\end{proof}
The index $n$ acts as an additional regularization parameter. For a given $\ai{e}{\vx}{\gamma} \in \mathcal{C}_0^\infty(\Omega)$, one can, for instance, consider the construction
$$
\ai{e}{\vx,n}{\gamma} := \sum_{m=0}^n \sum_{l=0}^{m} \distrib{\ai{e}{\vx}{\gamma}}{v_{ml}} v_{ml}.
$$
Please note that in this case, both sums in the representation \eqref{eq:kernel_gamma_n} of the corresponding kernel $\ai{\psi}{\vx}{\gamma,n}$ are finite as well.

An important issue with the limited-angle Radon transform is that the singular functions $v_{ml}$ contain the sum of oscillating functions at different frequencies, see \Cref{Sec:MathBackground}, i.e. low- and high-frequencies are scattered through the whole spectrum. This is counter-intuitive as often the smallest singular values tend to contain high-frequencies like in the case $\Phi=0$. In {\color{black} this case}, low-pass filters are efficient to dampen the smallest singular values. {\color{black} However, this is not true for $\Phi>0$.} This is why we add a \emph{spectral filter} to control the singular values. 

\begin{definition}\label{def:filter_spectrum}
Let $(F_\tau)_{\tau > 0}$ be a family of functions $F_\tau : (0,\infty) \to (0,\infty)$ such that 
\begin{enumerate}[label=(\roman*)]
    \item $\sup_{m,l} \vert F_\tau(\sigma_{ml}) \sigma_{ml}^{-1} \vert = c(\tau) < \infty$,
    \item $\lim_{\tau\to 0} F_\gamma(\sigma_{ml}) = 1$ for all $\sigma_{ml}$,
    \item $\vert F_\tau(\sigma_{ml})\vert $ is bounded for all $\gamma$ and $\sigma_{ml}$
\end{enumerate}
with $\sigma_{ml}$ the singular values of $\R_\Phi$, cf. \Cref{theo:SVD_R_Phi}. Then, $F_\tau$ is called \emph{spectral filter} for $\R_\Phi$.
\end{definition}
Examples include the classical Tikhonov-filter $F_\tau(\sigma)= \frac{\sigma^2}{\sigma^2 + \tau}$ or $F_\tau(\sigma) = \frac{\sigma}{\tau} \arctan \left( \frac{\tau}{\sigma} \right)$. 
With such a spectral filter, we now consider the following variant of the approximate inverse. 
 
\begin{theorem}\label{lemma:continuous_LARK}
    Let $\ai{e}{}{\gamma}$ be a  prescribed mollifier and let $(F_\tau)_{\tau>0}$ be a family of spectral filters for $\R_\Phi$.     Then, the operator $\S_{\gamma,\tau,n}: L_2(Z_\Phi;w^{-1}) \to L_2(\Omega)$ with $\S_{\gamma,\tau,n} g (\vx) := \langle g, \ai{\psi}{\vx}{\gamma,\tau,n}\rangle$ and kernel
    \begin{equation} \label{eq:continuous_LARK}
\psi_{\vx}^{\gamma,\tau,n} := \sum_{m=0}^n \sum_{l=0}^{m} \frac{F_\tau(\sigma_{ml})}{{\sigma}_{ml}} \scalarproduct{\ai{e}{\vx}{\gamma}}{v_{ml}} u_{ml} 
    \end{equation}
    is a regularization method for $\R_\Phi$. 
\end{theorem}

\emph{Remark:} We refer to $\psi_{\vx}^{\gamma,\tau,n}$ in \eqref{eq:continuous_LARK} as \emph{limited-angle reconstruction kernel (LARK)}.

\begin{proof}
We show that $\S_{\gamma,\tau,n} \to \R_\Phi^\dagger$ pointwise on $\mathcal{D}(\R_\Phi^\dagger)$ as $(\gamma,\tau,n) \to 0$. Then,  $(\S_{\gamma,\tau,n})_{\gamma,\tau,n> 0}$ is a regularization for $\R_\Phi^\dagger$ \cite{englbook}.
With the representation of $\ai{\psi}{\vx}{\gamma,\tau,n}$, it holds
$$\S_{\gamma,\tau,n}g (\vx) = \langle g, \ai{\psi}{\vx}{\gamma,\tau,n}\rangle = \left\langle \ai{e}{\vx}{\gamma}, \sum_{m=0}^n \sum_{l=0}^m F_\tau (\sigma_{ml}) \, \sigma_{ml}^{-1}\, \langle g,u_{ml}\rangle \, v_{ml}\right\rangle.$$
With $n\to \infty$, the second argument of the scalar product converges to 
$$\sum_{m=0}^\infty \sum_{l=0}^m F_\tau(\sigma_{ml}) \, \sigma_{ml}^{-1} \, \scalarproduct{g}{u_{ml}} \, v_{ml} =: \T_\tau g,$$
i.e. to a spectral filtering operator $\T_\tau$ applied to $g$. Since $(F_\tau)_{\tau>0}$ satisfies the conditions of definition \ref{def:filter_spectrum}, $(\T_\tau)_{\tau>0}$ is a regularization for $\R_\Phi^\dagger$, see \cite{louisbook}, i.e. in particular $\T_\tau g \rightarrow \R_\Phi^\dagger g$ for $g\in \mathcal{D}(\R_\Phi^\dagger)$ as $\tau\to 0$. 
Further, since $\ai{e}{}{\gamma}$ is a mollifier, it holds for any $h\in L_2(\Omega)$
$$\lim_{\gamma\to 0} \|\langle e^\gamma_\cdot, h \rangle - h\|_{L_2} = 0.$$
Together, we obtain for $g\in \mathcal{D}(\R_\Phi^\dagger)$
\begin{equation*}
    \underset{(\gamma,\tau,n) \to 0}{\lim} \| \S_{\gamma,\tau,n} g - \R_\Phi^\dagger g\| 
    \leq 
    \underset{(\gamma,\tau,n) \to 0}{\lim} \left(\| \S_{\gamma,\tau,n} g -\T_\tau g\| + \| \T_\tau g - \R_\Phi^\dagger g\|  \right) = 0. 
\end{equation*}
 
\end{proof}

\begin{remark}
    A straight-forward computation shows that LARK $\ai{\psi}{\vx}{\gamma,\tau,n}$ corresponds to the reconstruction kernel associated to the modified mollifier 
    \begin{equation}\label{eq:mollifier_gamma_tau}
    \ai{e}{\vx}{\gamma,\tau,n}(\vy) := \sum_{m=0}^n \sum_{l=0}^{m} F_\tau(\sigma_{ml}) \distrib{\ai{e}{\vx}{\gamma}}{v_{ml}} v_{ml}(\vy),
    \end{equation}
    \emph{i.e.} LARK corresponds to the solution of the auxiliary problem $\R_\Phi^* \psi = \ai{e}{\vx}{\gamma,\tau,n}$. Eq. \eqref{eq:mollifier_gamma_tau} can be seen as a model-driven construction. This construction might seem unnecessary at first glance since the filter $F_\tau$ suffices to define a regularization method for $\R_\Phi$. The problem is that the singular functions contain, as mentioned above and as seen in \Cref{theo:SVD_R_Phi}, a large range of frequencies. Thus, the regularization $\T_\tau$ does not control the smoothness of the solution but only the increase of $\sigma_{ml}^{-1}$ and therefore the norm of the regularization operator. The goal is not to cut-off the smallest $\sigma_{ml}$ as they embed the missing regions but to attenuate them. In order to control the higher frequencies in the solution -- particularly important with noisy data -- we therefore need to include the mollifier $e^\gamma$.     Finally, the truncation enforced by the parameter $n$ ensures the range condition. 
    Since $\ai{\psi}{\vx}{\gamma,\tau,n}$ corresponds to the solution of $\R_\Phi^* \psi = \ai{e}{\vx}{\gamma,\tau,n}$ with adapted mollifer, we simplify the notation - whenever unambiguous - by summarizing the regularization parameters $(\gamma,\tau,n)$ in a single parameter, denoted by $\gamma$.
    
\end{remark}
 
An advantage of the method of the approximate inverse, especially in tomography, is that invariance properties of the operator can significantly reduce the computation costs for determining reconstruction kernels for different reconstruction points $\vx$. As in the full-angle case, the limited-angle Radon transform satisfies
$
\mathcal{T}_1^\vx \R_\Phi^\star = \R_\Phi^\star \mathcal{T}_2^\vx
$
with $\mathcal{T}_1^\vx: f(\cdot)\mapsto f(\cdot-\vx)$ and $\mathcal{T}_2^\vx: g(s,\theta)\mapsto g(s-\vx^\top \theta,\theta)$. This property of $\R_\Phi$ can simplify the construction of the reconstruction kernel in the following way.
\begin{lemma} Let $e^\gamma$ be a mollifier of convolution type, \textit{i.e.} $\ai{e}{\vx}{\gamma}(\vy) = T_1^\vx \ai{e}{\mathbf{0}}{\gamma}(\vy)$. 
Then, $$
\ai{\psi}{\vx}{\gamma} =  \sum_{m=0}^{\infty}\sum_{l=0}^{m} \frac{\scalarproduct{\ai{e}{\mathbf{0}}{\gamma}}{{v}_{ml}}}{{\sigma}_{ml}} \left(T_2^\vx u_{ml}\right) .
$$
\end{lemma}
\begin{proof}
    This follows directly from the invariance properties of $\R_\Phi$, see also \cite{louis2006development}.
\end{proof}
However, this property has to be treated with caution in real applications. Replacing the direct computation of $\ai{\psi}{\vx}{\gamma,\tau,n}$ by evaluating $T_2^\vx \ai{\psi}{\mathbf{0}}{\gamma,\tau,n}$ with $\ai{\psi}{\mathbf{0}}{\gamma,\tau,n}$ precomputed on a certain grid will require interpolation steps. These imply errors which are harmful for such an exponentially ill-posed problem. Thus, exploiting the invariances would come at the cost of an over-stabilization of the reconstruction process which we choose to avoid.

{\color{black} \begin{remark}
    Our approach can be extended in a straight-forward manner to enable the \emph{feature reconstruction} by replacing $e^\gamma_\vx$ with $\mathcal{L}^* e^\gamma_\vx$, where $\mathcal{L}$ is the desired feature operator, see \cite{louis2011feature}.  
\end{remark}}
 
\subsection{Handling data errors under severe ill-posedness}\label{sec:denoising}

When dealing with real data, one has to face errors in the measurement and, more generally, imperfections in the forward model. The discrepancy between such data $g^\delta$ and $g = \R_\Phi f$ will impact the whole frequency range and, particularly high frequencies. Given the fast decay of the singular values and the frequency richness of the singular functions, the reconstruction $\S_\gamma g^\delta$ is expected to produce large artefacts arising from the singular functions corresponding to the smallest singular values like shown in \Cref{fig:singular_and_kernel}(d). In order to handle this challenge, one has to further constrain the solution space. A typical approach in imaging applications is to impose a small total variation (noted $\mathrm{TV}(\cdot)$) on the solution, resulting in solving
\begin{equation}\label{eq:TVstandard}
\min\limits_f \ \frac12 \|\R_\Phi f - g^\delta\|_{L_2}^2  +\lambda \, \mathrm{TV}(f), \quad \lambda>0
\end{equation}
which tends to preserve edges but remains limited for large $\Phi$.

This formulation, however, does not allow to use the reconstruction kernel built by the approximate inverse and its benefits. As we will see in the later sections, the reconstruction kernel changes the nature of the artefacts in the reconstructed image - instead of the characteristic streak artefacts and missing features, the singular functions to extremely small singular values degrade the image. In order to combine the approximate inverse with constraints on the solution space, we propose to define the following solver
\begin{equation}\label{eq:reg-AI-discrete}
\D_\lambda (g^\delta) \in \argmin\limits_g \frac12 \|g - g^\delta\|^2 +\lambda \, \P\circ\S_\gamma(g), \end{equation}
with $\S_\gamma$ as defined in \eqref{eq:def_S_gamma} and $\P$ a suited penalty functional compensating for the smoothing of the missing regions. Potential choices for the regularizer $\P$ include classical total-variation or more general non-linear isotropic diffusion functionals \cite{hahn24nid} which have the advantage of allowing edge sharpening. 

The operator $\mathcal{D}_\lambda$ can be seen as a pre-processing denoising step on the data $g^\delta$ enforcing data fidelity and edge-preservation in the final reconstruction $f_\gamma^\lambda:=\S_\gamma(\D_\lambda (g^\delta))$. We refer to $f_\gamma^\lambda$ as solution via the \textit{constrained limited-angle reconstruction kernel} (CLARK).\\
 
The following theorem shows that the existence and uniqueness of a minimizer $D_\lambda(g^\delta)$ can be ensured by the choice of the penalty term $\mathcal{P}$.

\begin{theorem}\label{thm:existence-clark}
    Let $\mathcal{P}\colon L_2(\Omega)\to [0,\infty]$ be weakly sequentially lower semi-continuous s.t. there exists $f\in \range(S_\gamma) $ with $\mathcal{P}(f)< \infty$.
    Then, a minimizer $\D_\lambda(g^\delta)$ as in equation \eqref{eq:reg-AI-discrete} exists. If, in addition, $\mathcal{P}$ is convex, then this minimizer is unique.
\end{theorem}

\begin{proof} According to its construction, $S_\gamma$ is a linear and bounded operator, cf. \cite{Schusterbook}. Thus, under the given conditions on $\mathcal{P}$, it holds that
    \begin{itemize}
        \item $\mathcal{P}\circ S_\gamma$ is weakly sequentially lower semi-continuous
        \item there exists $g$ s.t. $\mathcal{P}(S_\gamma(g))<\infty$
        \item $\mathcal{P}\circ S_\gamma$ is convex whenever $\mathcal{P}$ is convex.
    \end{itemize}
    Furthermore, the level sets $\{g\in L_2\colon \mathcal{Q}(g)\leq t\}$ of the functional $       \mathcal{Q}(g):=\frac{1}{2}\norm[g-g^\delta]^2+\lambda\cdot\mathcal{P}\circ S_\gamma(g)$
                        are bounded due to the first term and $\mathcal{P}\geq 0$. Hence, with $L_2$ being a reflexive space, said level sets are even weakly sequentially compact.  
    The existence and uniqueness of solutions now follow from standard arguments in (convex) variational optimization, see e.g. \cite{scherzer00}, Theorem 3.22 or \cite{bredies11}, Theorem 6.31.
\end{proof}

Please note that the asserted conditions are, e.g., satisfied by choosing $\mathcal{P}$ as total variation or as the functionals from \cite{hahn24nid}.\\
 
\section{Towards the use on real data}\label{sec:semi_discrete-AI}

When dealing with real data, the model needs to be adapted 
to the measurement sampling. We denote by $\Xi_m: L_2(Z_\Phi;w^{-1}) \to \RR^m$ an \textit{observation operator} such that real data $\gvec \in \RR^m$ satisfy the semi-discrete forward model
\begin{equation} \label{eq:semidiscreteFM}
\gvec = \Xi_m \R_\Phi f \qquad \text{for } f\in L_2(\Omega).
\end{equation}

Due to the numerical challenges when working with the SVD of the limited-angle Radon transform, see Section \ref{Sec:MathBackground}, and since the semi-discrete operator is unbounded \cite{rieder00}, we propose, in an unconventional way, to instead solve the inverse problem directly in its discrete form, which, in turn however, comes with its own challenges. \\ 

If the searched-for function $f$ stems from interpolation of a discrete version $\fvec \in \RR^n$, i.e. $f=\Pi_n \fvec$ with \textit{interpolation operator} $\Pi_n: \RR^n \to L_2(\Omega)$, then the forward model in \eqref{eq:semidiscreteFM} further reduces to a simple \textit{projection matrix} under the form
\begin{equation}\label{eq:def_A_Phi}
A_{\Phi,\Xi_m,\Pi_n} \fvec := \Xi_m \R_\Phi \Pi_n \fvec, \qquad \text{with } A_{\Phi,\Xi_m,\Pi_n} \in\RR^{m\times n}.
\end{equation}
In the following, we consider the straightforward construction of the projection matrix $A_{\Phi,\Xi_m,\Pi_n}$ with the standard parallel beam data sampling and object sampling in a given basis, i.e.
\begin{itemize}
    \item $\Xi_m$ samples a continuous function $g(s,\theta)$ in a measurement grid $((s,\theta)_j)_{j=1,\ldots,m}$, which is uniform with respect to $s$ and $\theta$, respectively,            \item $\Pi_n$ stands for the interpolation over a grid with centers $(\vx_i)_{i = (1,\ldots,n)}\subset \RR^2$ and basis function $\phi \in L_2(\Omega)$, i.e. for a given vector $\fvec\in\mathbb{R}^n$, it holds 
\begin{equation} \label{eq:interpol_op} \Pi_n \fvec = \sum_{i=1}^n {\fvec}_i \, \phi (\cdot - \vx_i).\end{equation}
If $f$ and $\phi$ can be evaluated pointwise, the coefficient vector $\fvec$ is determined such that $f(\vx_j)=\Pi_n \fvec (\vx_j)$ for all $j=1,\dots,n$. 
\end{itemize}
Thus, the entries of the matrix $A_{\Phi,\Xi_m,\Pi_n}$ are given by
$$
(A_{\Phi,\Xi_m,\Pi_n})_{ji} = \int_\Omega \phi(\vx-\vx_i) \delta(s_j - \vx \cdot \theta_j) \mathrm{d}\vx.
$$
When there is no ambiguity and for the sake of readibility, we shall replace $A_{\Phi,\Xi_m,\Pi_n}$ by $A_\Phi$. We will speak of the discrete case when applying the discrete LARK/CLARK on discrete data and of the semi-discrete case when on real/analytic data.

\subsection{Computation of the discrete reconstruction kernel} 

For the reconstruction step, we choose a mollifier $e^\gamma$ as well as discrete reconstruction points $(\vz_l)_{l=1,\dots,r}\subset \mathbb{R}^2$. 
These can potentially differ from the centers $(\vx_i)_{i=1,\dots,n}$ used to build the operator $\Pi_n$. This gives us the flexibility to use two different representations for $f\in L_2(\Omega)$ -- one for constructing the projection matrix $A_\Phi$ and another one for the actual reconstruction step. For the latter, we ideally want to work with the classical pixel grid since this facilitates for instance the computation of gradients and allows a direct combination with our additional constraint on the solution space, cf. Section \ref{sec:denoising}. However, this might not sufficiently well approximate $f$ in order to explain the semi-discrete data $\gvec$. Using a different grid for representing the discrete projection matrix will enable us to reduce the overall reconstruction error, introduced when applying the discrete reconstruction kernel to semi-discrete data $\gvec$ later. \\
 
Once the reconstruction points $(\vz_l)_{l=1,\dots,r}$ are fixed, the goal is to recover from $\gvec$ a vector $\fvec^\gamma \in \mathbb{R}^r$ approximating $( \langle f, e^\gamma_{\vz_1}\rangle, \dots, \langle f,e^\gamma_{\vz_n}\rangle)^\top =: \fvec_{\vz}^\gamma \in\mathbb{R}^r$. To this end, we define 
$$
\fvec^\gamma := \Big( \langle \Pi_n \fvec, e^\gamma_{\vz_j} \rangle_{L_2} \Big)_{j=1,\dots,r}.
$$
With the representation of $\Pi_n$ in \eqref{eq:interpol_op}, this corresponds to $\fvec^\gamma = E^\gamma \fvec$ where
\begin{equation*}
E^\gamma \in \mathbb{R}^{r \times n} \ \text{ with } \ (E^\gamma)_{ji} := \int_\Omega \phi(\vx-\vx_i)\, e^\gamma_{\vz_j} (\vx) \,\mathrm{d}\vx. \label{eq:semi_disc-moll}
\end{equation*}
We now propose to compute the reconstruction kernel with respect to the discrete projection matrix $A_{\Phi}$ by solving the auxiliary problem 
\begin{equation} \label{eq:auxprobldiscrete} 
A_{\Phi}^\top \Psi^\gamma = (E^\gamma)^\top. 
\end{equation}

Following our analysis in Section \Cref{sec:continuousLARK} for the continuous case, we compute $\Psi^\gamma$ via the singular value decomposition 
$(\sigma_i, \mathrm{u}_i,\mathrm{v}_i)_{i=1,\ldots,n}$ of the matrix $A_\Phi^\top \in \RR^{n  \times m}$, $m \geq n$, i.e. $A_\Phi^\top \mathrm{u}_i = \sigma_i \mathrm{v}_i$ with $\mathrm{u}_i \in \RR^m$ and $\mathrm{v}_i \in \RR^n$. In matrix form, this reads $A_\Phi^\top = VSU^\top$ with  $V\in \RR^{n \times n}$, $S = \mathrm{diag}((\sigma_i)_i)\in \RR^{n \times n}$ and $U\in \RR^{m \times n}$. Using the SVD, the solution to \eqref{eq:auxprobldiscrete} is simply \begin{align}\label{eq:FilterComputationReckernel}
\Psi^\gamma = U S^{-1} V^\top (E^\gamma)^\top. 
\end{align}    
As discussed in Section 2, $\R_\Phi$ spreads the frequencies of the preimage $f$ over the complete spectrum, and this property is inherited by the discrete forward operator $A_\Phi$. With the strategy developed in \Cref{sec:continuousLARK}, the LARK for the discrete setting reads \begin{align}\label{eq:FilterComputationReckerneltau}
    \Psi^{\gamma,\tau} := U \Sigma_\tau V^\top (E^\gamma)^\top, \qquad \text{where we set } \Sigma_\tau := \text{diag}\left( \left( F_\tau (\sigma_i) \right)_{i} \right) S^{-1}
    \end{align}
    with spectral Filter $F_\tau$. Please note that the additional regularization parameter $n$ does not appear here in the discrete case since the range $\mathrm{Ran}(A_\Phi)$ is closed and hence, no range condition needs to be imposed. As in the continuous case, $\Psi^{\gamma,\tau}$ corresponds to the solution of the auxiliary problem with modified right-hand side, \textit{i.e.} 
    $$
    A_{\Phi}^\top \Psi^{\gamma,\tau}  = E^{\gamma,\tau}:= E^\gamma (V\Sigma_\tau S V^\top).
    $$
    We note that $E^{\gamma,\tau}$ defines a \textit{nascent} matrix as it tends to the identity matrix when $\gamma,\tau \to 0$. To simplify the notation we omit the second regularization parameter $\tau$ in the rest of the manuscript (if there is no ambiguity). \\%focus on equation \eqref{eq:FilterComputationReckernel} throughout the theoretical part, to keep the notation simple.
 
\begin{figure}[h!]\centering
\begin{subfigure}{0.27\linewidth}
    \includegraphics[width=\linewidth]{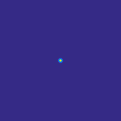}
    \caption{}
\end{subfigure}
\begin{subfigure}{0.27\linewidth}
    \includegraphics[width=\linewidth]{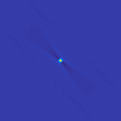}
    \caption{}
\end{subfigure}
\begin{subfigure}{0.27\linewidth}
    \includegraphics[width=\linewidth]{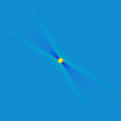}
    \caption{}
\end{subfigure}
    \caption{Visualization of the different mollifiers for the central pixel on a $N\times N$ grid with $N=121$ and $\Phi=30^\circ$ for: 
    (a) $E^\gamma$, (b) $E^{\gamma,\tau}$, (c) $\arctan \left(N \cdot E^{\gamma,\tau}\right)$ (for a better visualization) where $\gamma = 1/N$ and $\tau = 5\sigma_n$.}
    \label{fig:placeholder}
\end{figure}

\subsection{Ideal case: Exact interpolation in the pixel basis}\label{sec:discrete}

In this subsection, we assume that $f$ is exactly represented by $\Pi_n$, which means that $f = \Pi_n \fvec$. In this case, using the proposed reconstruction kernel $\Psi^\gamma$ allows in fact to recover $\fvec^\gamma$ from the data $\gvec$, as the following theorem shows. 
 
\begin{theorem}
    Let $A_{\Phi,\Xi_m,\Pi_n}, \, E^\gamma$ and $\Psi^\gamma$ as above in Section 3.1. If $f$ can be exactly represented by $\Pi_n$, it holds
    $$(\Psi^\gamma)^\top \gvec = \fvec^\gamma.$$
\end{theorem}

\begin{proof}
    With $f=\Pi_n \fvec$, it follows $\gvec = \Xi_m \R_\Phi f = \Xi_m \R_\Phi \Pi_n \fvec = A_\Phi \fvec$. Thus, with the auxiliary problem \eqref{eq:auxprobldiscrete}, it holds
    $$(\Psi^\gamma)^\top \gvec = (\Psi^\gamma)^\top A_\Phi \fvec = \left(A_\Phi^\top \Psi^\gamma\right)^\top \fvec= (E^\gamma)^\top \fvec = \fvec^\gamma.$$
\end{proof}
 
Now we further consider, that $f$ can be exactly represented in a regular pixel basis, i.e. 
$$
f = \sum_{i=1}^n \fvec_i \phi (\cdot -\vx_i)
\qquad
\text{with } 
\phi(\vx) = 
\begin{cases}
    1 &\text{if } \vx \in \big[-\frac{h}{2},\frac{h}{2}\big) \times \big[-\frac{h}{2},\frac{h}{2}\big),\\
    0 &\text{otherwise,}
\end{cases}
$$
$h>0$ and where $\vx_i$ denotes the center of the $i$-th pixel. Here the regular pixel grid is a square, \textit{i.e.} $h = 1/\sqrt{n}$. In particular, each vector $\vx\in\RR^2$ belongs to exactly one pixel, 
say the one associated to $\vx_k$, i.e. $\phi(\vx_k-\vx_i)=\delta_{ki}$, and hence, $\fvec = (f(\vx_1),\dots,f(\vx_n))^\top \in \RR^n$. Using the exact interpolation property $f=\Pi_n \fvec$, it holds with $P_i$ denoting the pixel associated to the center $\vx_i$ $$\langle f, e^\gamma_{\vx_k}\rangle = \sum_{i=1}^n f(\vx_i)\, \int_{P_i} e^\gamma_{\vx_k} (\vx)\,\mathrm{d}\vx  .$$
Thus, by choosing $(\vx_i)_{i=1,\dots,n}$ as reconstruction points, we can recover
$ \fvec^\gamma := E^\gamma \fvec$
with symmetric $E^\gamma \in \RR^{n\times n}$, $E^\gamma_{kl} := e^\gamma_{\vx_k}(\vx_l)$ from the discrete data $\gvec =A_\Phi \fvec$.
 
Besides the approximate inverse reconstruction $\text{f}^\gamma = (\Psi^\gamma)^\top \gvec$, we also consider its constrained counterpart $\text{f}^{\gamma,\lambda} = (\Psi^\gamma)^\top \mathcal{D}_\lambda(\gvec)$ with $\mathcal{D}_\lambda$ as in eq. \eqref{eq:reg-AI-discrete}.  
 
\paragraph{Numerical results} 
First, we want to visualize our reconstruction kernel which we computed using the spectral filter $F_\tau(\sigma) = \frac{\sigma}{\tau} \arctan \left( \frac{\tau}{\sigma} \right)$ and the Gaussian mollifer. 
Figure \ref{fig:LARKs_pixel_different_Phi} shows the reconstruction kernel LARK computed for the central pixel as reconstruction point for different $\Phi$ in the $(\theta,s)$-coordinate system. One can see the oscillatory nature of the LARK which increases with $\Phi$ and which is expected given the severe ill-posedness of the limited-angle problem. This behavior, {\color{black}more precisely the exponential decay of the singular values of $A_\Phi$ for different $\Phi$} is also validated in Fig. \ref{fig:singular_values_101_pixel}. {\color{black} Like expressed in \cite{Natterer86}, the exponential decay sets in earlier as $\Phi$ increases.}
 
\begin{figure}[!h]
    \centering
    \includegraphics[width=\linewidth]{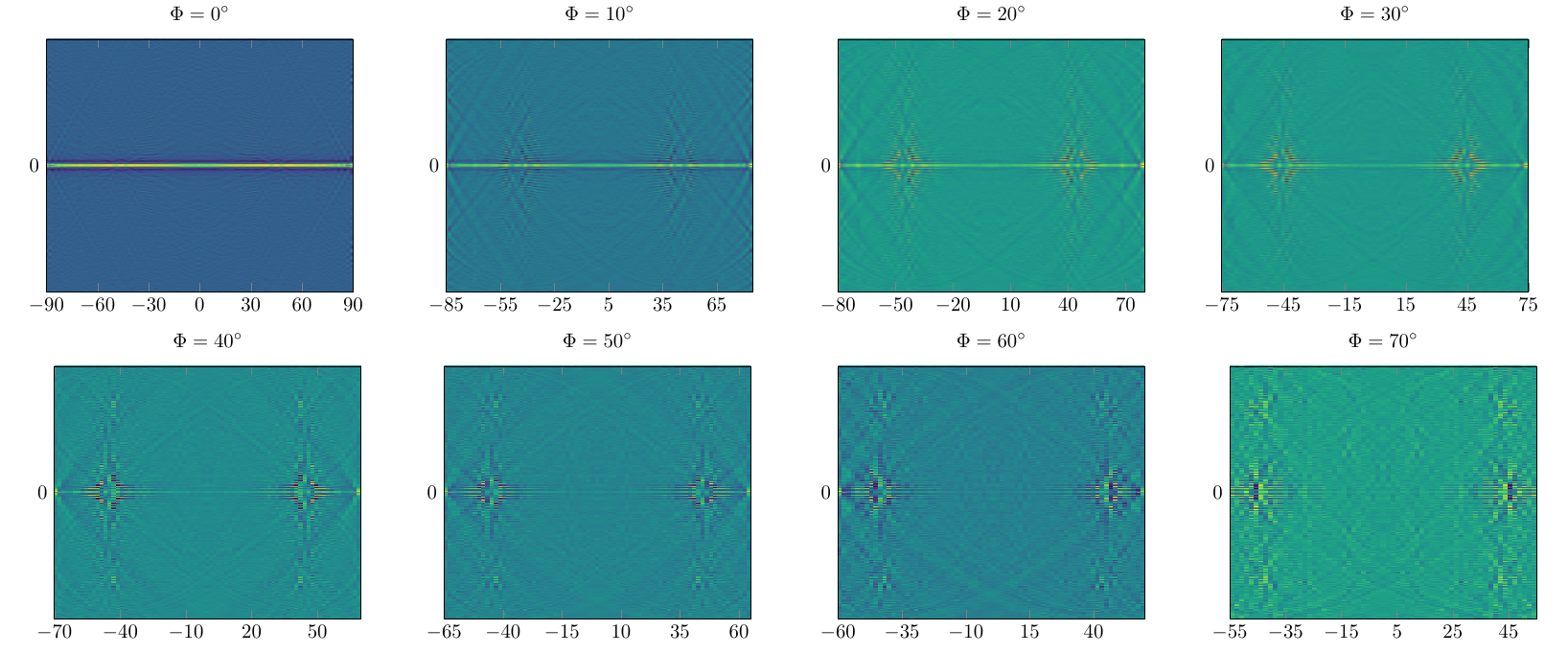}
    \caption{Evolution of the reconstruction kernel LARK      at reconstruction point $\vz=0$ as $\Phi$ increases.     }
    \label{fig:LARKs_pixel_different_Phi}
\end{figure}

\begin{figure}[h!]
    \centering
    \includegraphics[width=0.8\linewidth]{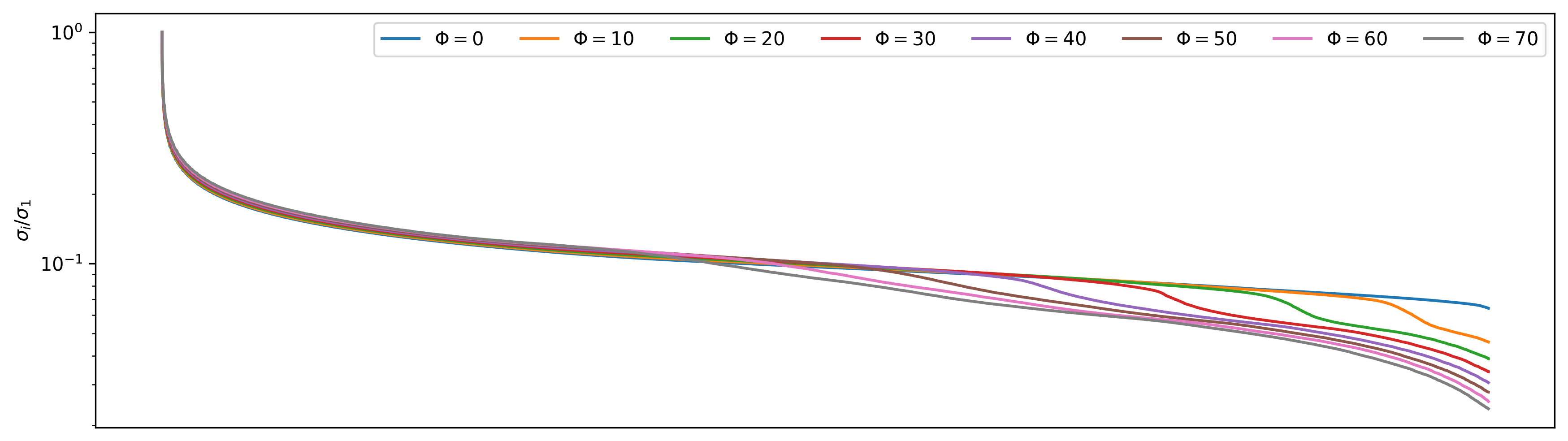}
    \caption{Evolution of the singular values of $A_\Phi$ for the pixel grid and for different $\Phi$.}
    \label{fig:singular_values_101_pixel}
\end{figure}
 
To validate our reconstruction approach, we generate discrete data for the Shepp-Logan-Phantom with $n=201\times 201$, $m=512\times 200$ in case of 180 degrees and for different limited-angle scenarios, i.e. for different values of $\Phi$. To simulate measurement noise, we add samples of normally distributed random variables to the discrete data $\gvec$, resulting noisy data $\gvec^{\epsilon}$ with different noise levels ranging from $0.05\%$ to $10\%$, see Table \ref{tab:data-snr}. The noise levels refer to the relative data error $\frac{\norm[\gvec^{\delta}-\gvec]}{\norm[\gvec]}$. Table \ref{tab:data-snr} also states the corresponding SNR values. Please note that the specific perturbation vector was computed such that the given values hold for all $\Phi$ and independent of the dimension of $\gvec$. 	\begin{table}[H]
		\centering
		\begin{tabular}{c|c|c|c|c|c}
		0.05 \%&0.10 \%&1.00 \%&2.00 \%&5.00 \%&10.00 \%\\\hline
		66.02&60.00&40.00&33.98&26.03&20.04
		\end{tabular}
		\caption{\label{tab:data-snr}SNR values for different noise levels}
	\end{table}

We build the mollifier $E^\gamma$ from the Gaussian function $\ai{e}{\vz}{\gamma}(\vx) = \frac{1}{2\pi \gamma} \exp\left(-\frac{\|\vx-\vz\|^2}{2\gamma}\right)$ and compute the respective reconstruction kernel $\Psi^\gamma$ via \eqref{eq:FilterComputationReckernel}. To obtain the constrained solution $\fvec^{\gamma,\lambda}$, we use in \eqref{eq:reg-AI-discrete} the total variation regularization, more precisely the continuously differentiable approximation proposed by Acar and Vogel in \cite{acar1994}, such that we can solve the minimization problem with gradient descent. 
 
\renewcommand{\tablepath}{figures/stability_tests/tables}
We compare our approach using the limited-angle reconstruction kernel (LARK) as well as the constrained counterpart  (CLARK) to the standard filtered backprojection algorithm (FBP) with Shepp-Logan filter and to the standard minimization approach \eqref{eq:TVstandard} with total variation (TV). Figure \ref{fig:compNoiseLevel} shows the respective results for a missing angular range of $30^\circ$ and for the different noise levels from Table \ref{tab:data-snr}. 
 
  \begin{figure}
        \centering
                \includegraphics[width=\textwidth]{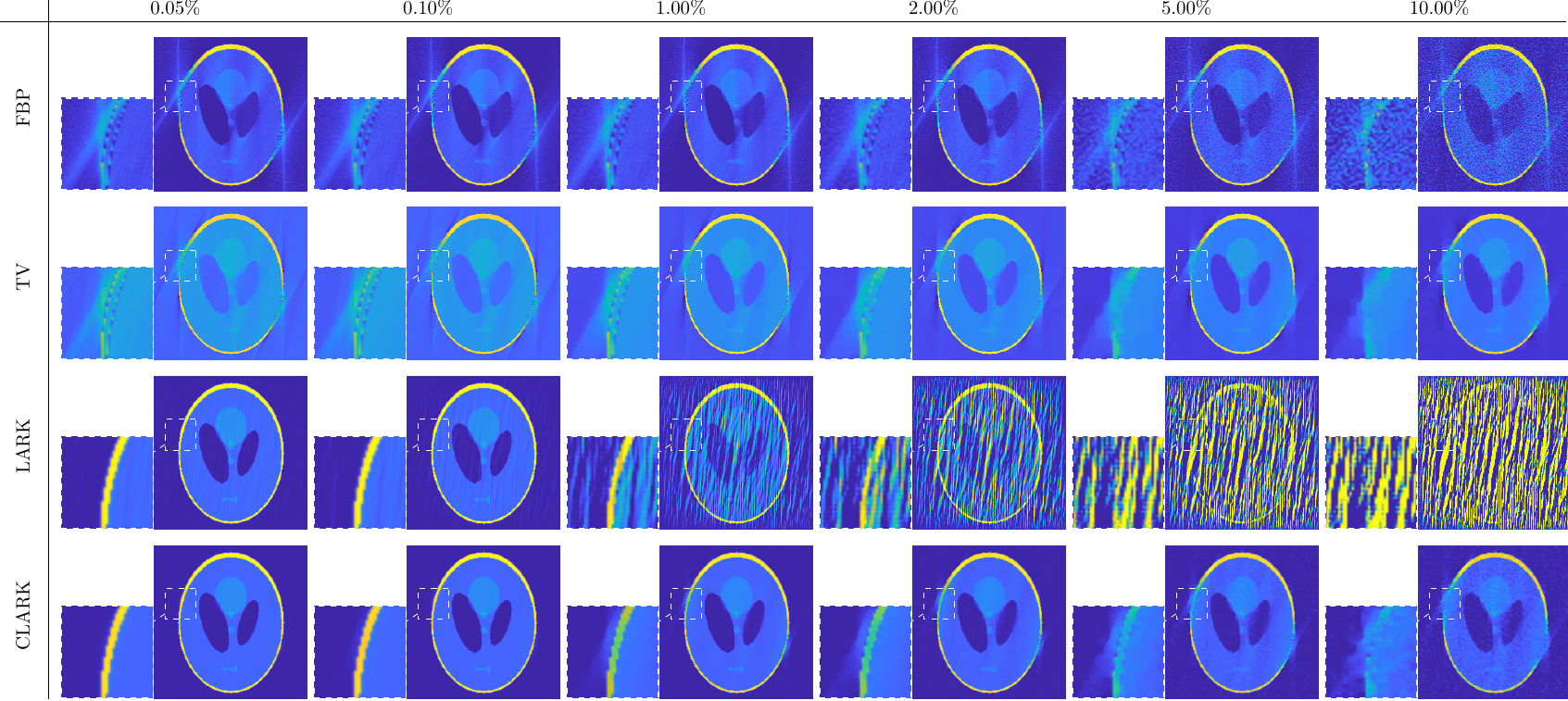}
        \caption{Comparison of different noise levels for $30^\circ$ missing angles}
        \label{fig:compNoiseLevel}
    \end{figure}
 
The FBP performs as expected and known from the literature: The object features associated to the missing angular range are not recovered and instead the characteristic streak artefacts arise. Standard TV produces sharp edges while suppressing the noise but also fails to recover the features associated to the $30^\circ$ missing angular range. Using our reconstruction kernel, however, we obtain for small noise level a good approximation to the ground truth - despite the $30^\circ$ missing angles: All features are correctly reconstructed, even those associated to the missing angle, without the characteristic limited-angle streak artefacts. 

However, as the noise level increases, we observe for LARK, i.e. without the additional constraint on the solution space, a new type of artefacts: For $1\%$ noise, the interior structure of the object is already significantly overlaid with wave-type artefacts which stem from the small singular values, see \Cref{fig:singular_and_kernel} d). With further increasing noise level, these artefacts become quickly more prominent covering the object beyond recognition. These artefacts are not a consequence of the discretization -- they have been observed and analysed in the continuous model as well, see \cite{louis1986incomplete} and \cite{peres1979} . 

In contrast, our constrained approach CLARK provides stable results even for high noise levels: The artefacts arising from the small singular values are successfully suppressed and, furthermore, the object features associated to the missing angles are correctly reconstructed. As to be expected, the higher the noise level, the larger the regularization parameter needs to be chosen which results in overall smoother reconstruction results and a gradual return of the limited-angle artefacts. Nevertheless, for $5\%$ noise, the "{}missing"{} features still get more reliably reconstructed by our approach than with standard FBP or TV. \\
 
Next, we illustrate how our approach performs with respect to $\Phi$, i.e. with respect to the extent of the missing angular range. The larger $\Phi$, the more severe becomes the ill-posedness. Therefore, for a fair comparison, the considered noise level varies for each $\Phi$. The respective values are given in the header of Figure \ref{fig:compAngularComp}. For each pair of $\Phi$ and noise level, we then show reconstruction results from standard FBP (first row), classic TV (second row), LARK (third row) and CLARK (fourth row), respectively. While the standard TV approach can handle a small missing angular range of $10^\circ$, it fails to recover all object features as the latter increases. Our approach, however, still succeeds in correctly recovering all object features, even those associated to the missing angles. Even for data missing on a large angular range, the contours are closed correctly (also those invisible in the FBP- and TV-reconstruction). This shows that the proposed method in fact stably inverts the limited-angle forward operator in the discrete scenario. However, we also observe that the tolerable noise level decreases as $\Phi$ increases. For 10 degrees missing, 2\% noise can be handled very well by CLARK, (and LARK apart from the wavetype artefacts stemming from the singular functions to small singular values). For $70^\circ$ missing, we could turn up the noise level until 0,1\% to get such a good reconstruction. For higher noise levels, LARK and CLARK encounter increasing difficulty in correctly recovering the missing features. This is again a direct consequence of the severely ill-posed nature of the limited-angle problem. 
 
                            \begin{figure}
        \centering
                \includegraphics[width=\linewidth]{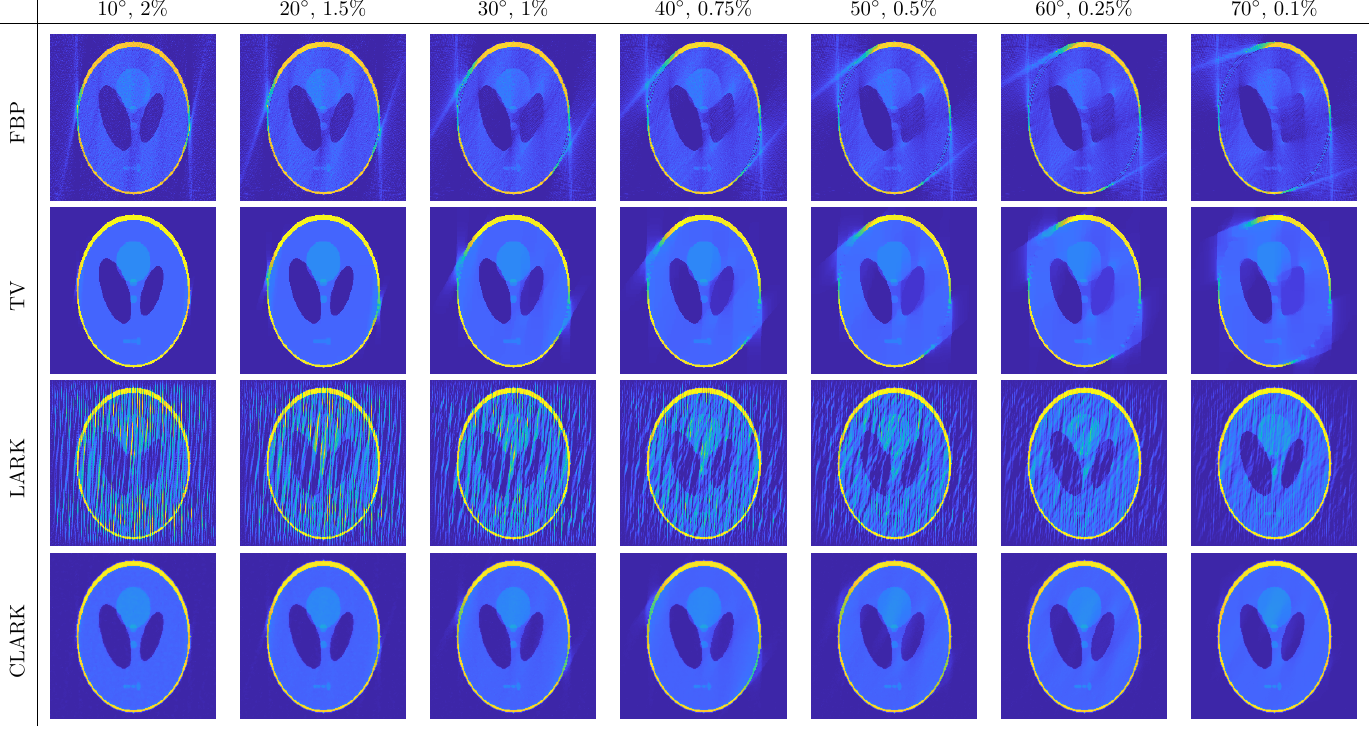}   
        \caption{Comparison for different angular ranges $\Phi$ and noise levels}
        \label{fig:compAngularComp}
    \end{figure}
 
\subsection{The real-data case: a representation issue}
\label{sec:semi_discrete}

In general, a given $f\in L_2(\Omega)$ will not be exactly represented by the interpolation operator $\Pi_n$. In this case, real measurements or analytical data $\gvec = \Xi_m\R_\Phi f$ will differ slightly in nature from the synthetic data $\Xi_m \R_\Phi \Pi_n \fvec$ due to the non-vanishing approximation error $\|f-\Pi_n \fvec\|>0$. Nevertheless, we still want to use a discrete reconstruction kernel $\Psi^\gamma$ obtained from solving the auxiliary problem with a discrete projection matrix $A_{\Phi,\Xi_m,\Pi_n}$ to circumvent the problems arising when dealing with the semi-discrete forward operator.\\%, see discussion before. \\

While such a small modelling error does not play a significant role in the full-data scenario, this changes in the severely ill-posed limited-angle case. This is illustrated in Figure \ref{fig:AluSi_representation_issue} for both real data, courtesy of the Fraunhofer Institute for non-destructive testing (IzfP) in Saarbrücken, as well as analytic data for the Shepp-Logan Phantom (please note that the Radon transform of geometric shapes such as ellipses can be stated exactly). Regarding the IzfP data, a silicium cylinder circled by aluminum was scanned by CT, resulting in the real data $\gvec$ for $\Phi = 0$. A standard FBP reconstruction then leads to a fair reconstruction of the object $\fvec$. 
We then apply the LARK $\Psi_\Phi^\gamma$, as built in the previous section, to the real / analytical data with limited angle, noted here $\gvec_{\Phi}$ (second and fourth rows). We can see that the real data scenario and the analytic one fail gradually as $\Phi$ increases. Looking more closely to the data and particularly to the difference between them, one can notice structural differences in-between. One can interpret it in the two following equivalent ways.
\begin{itemize}
    \item \textit{In the $\fvec$-space}, generated by the pixel grid, a disk is represented as a collection of pixels with harsh and inaccurate edges. This construction actually strongly constrains the space of $\fvec$ and therefore leads to a less ill-conditioned problem $A_\Phi \fvec = \gvec$ to solve. 
    \item \textit{In the $\gvec$-space}, many \textit{oscillations} embed the boundaries of the pixels. Real and analytical data do not contain such features, as it stands for a sampled version of the continuous operator $\R_\Phi$. To visualize this difference, we consider the following coefficients:
    $$
    \varrho_n^\Phi = \vert S^{-1} U(\gvec_\Phi - A_\Phi \fvec)\vert.
    $$
    These coefficients would then be multiplied by $V$ to get the LARK of $(\gvec_\Phi - A_\Phi \fvec)$ for $\gamma \to 0$. They are depicted in Figure \ref{fig:AluSi_representation_issue} for $\Phi=60^\circ$ and for both real and analytical data. The different \textit{encoding} of the contours between $A_\Phi$ and $\Xi_m \R_\Phi$ results in large coefficients $\varrho_n^\Phi$ especially for the largest inverse singular values.  
    Since the missing details are precisely contained in the associated singular vectors in $V$, the LARK, as defined in the previous section, is then expected to fail when applied to semi-discrete data (real or analytical).
\end{itemize}    

\begin{figure}[!h]
    \centering
    \includegraphics[width=\linewidth]{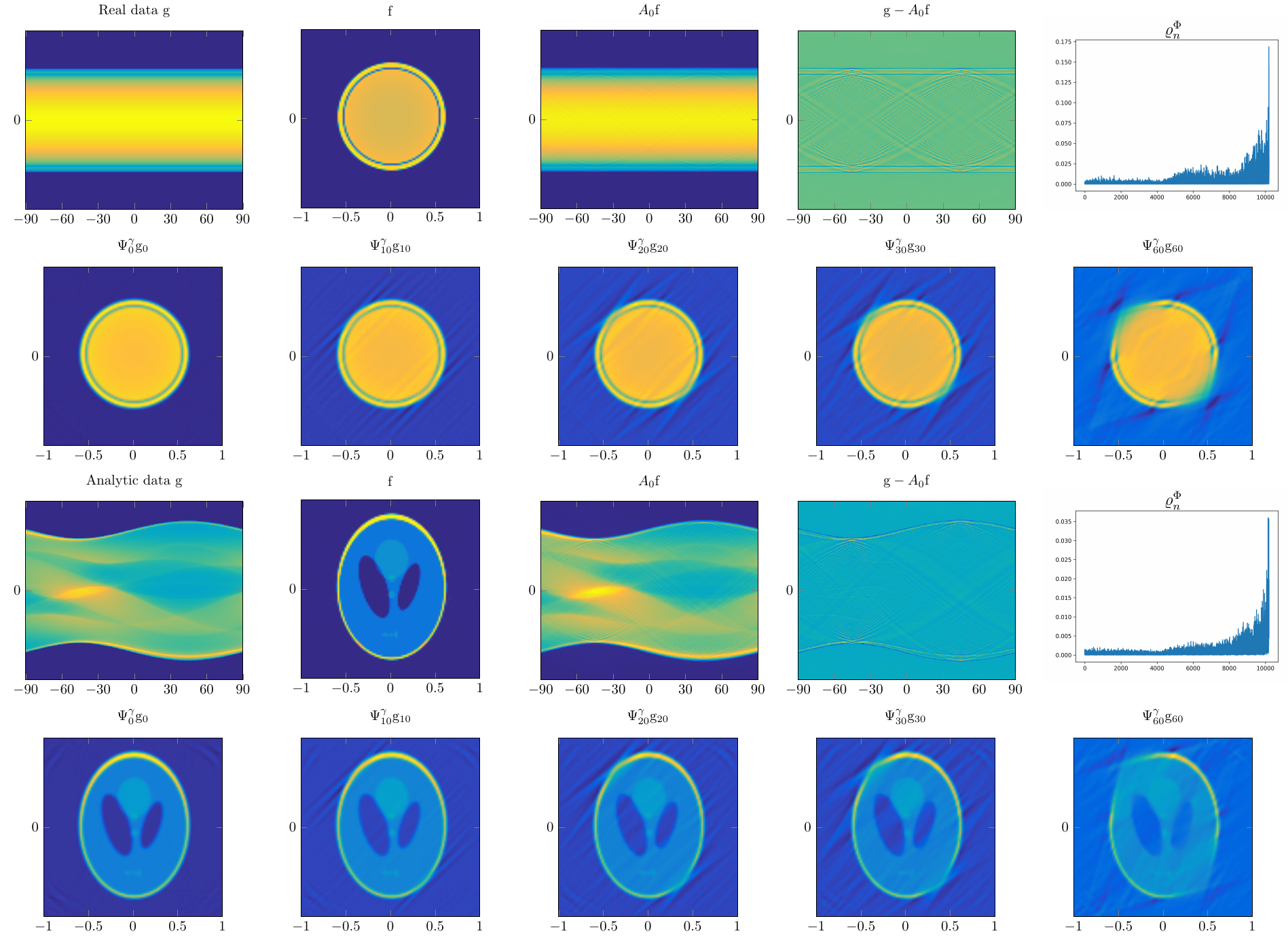}
    \caption{Application of the LARK defined on a rectangular pixel grid on real (first two rows) and analytical data (last two rows).}
    \label{fig:AluSi_representation_issue}
\end{figure}

In conclusion, the standard construction of $A_\Phi$ on a pixel grid delivers a too simplistic and too \textit{well-conditioned} forward model which does not suffice to account for the complexity of the limited-angle case. We therefore need to make $A_\Phi \fvec = \gvec$ more difficult to solve in order to apply the (C)LARK approach to real and analytic data. 

The key when approximating $\Xi_m \R_\Phi$ by $A_\Phi$ relies upon the construction of $\Pi_n$ which depends on the grid $(\vx_i)_{i=1\ldots n}$ and the function $\phi$. In approximation theory, interpolation operators $\Pi_n$ as in introduced in (9) are well-known and studied in the context of reproducing kernel Hilbert spaces and scattered data approximation. In particular, interpolation by a suitable \emph{radial basis function} $\phi$ leads - at least for certain $f$ - to pointwise error estimates of the form
\begin{equation} \label{eq:interpolerrorRKHS} \| f - \Pi_n \fvec\|_\infty \leq C\, P_{\phi}(h) \, \|f\|_\phi \end{equation}
with a function $P_\phi$ depending on the so-called \emph{grid density measure} (or \emph{fill distance}) $h$, see \eqref{eq:griddensitymeasure}, a constant $C$ independent of $f$ and $h$ and a norm $\norm[\cdot]_\phi$.  
Depending on the basis function, the asymptotic behavior of $P_\phi$ is given by $\mathcal{O}(h^\mu)$ with a power $\mu$ depending on $\phi$, see e.g. \cite{narcowich2003refined} and references therein, i.e. the interpolation error converges to zero for $h\to 0$.  For infinitely smooth radial basis functions like the Gaussian, convergence is even exponential \cite{chen2005scientific}. 
 
After a brief introduction into this general theory, we will make use of this type of estimates to derive and analyze error bounds for the reconstruction error
$$ 
\left\| (\Psi^\gamma)^\top \gvec - \fvec^\gamma_\vz\right\|_{\mathbb{R}^{r}} ,
$$
i.e. for the error we obtain when applying the reconstruction kernel $\Psi^\gamma$ from the fully-discrete case with interpolation operator $\Pi_n$ to the (true) semi-discrete data $\gvec$. 
{\color{black} \begin{remark}
    {\color{black} Reproducing kernel Hilbert spaces have been used in the context of inverse problems in the past.     In \cite{krebs09} the authors examined general error estimates and parameter choices for Thikononv regularizations of general inverse problems.  Specifically for the Radon transform they were studied in the context of algebraic reconstruction techniques in \cite{DeMarchi2018} for noise-free data, while  in \cite{wang21} a generalized approach was considered, including Thikonov regularization and noisy data. In both cases the limitations in the data, however, came from \emph{scattering}, i.e. the data were distributed irregularly with some missing lines.     } 
\end{remark}}
 
\subsubsection{An introduction to interpolation with radial basis functions} 
 
\subsubsection*{General regular grids}
 
As indicated by \eqref{eq:interpolerrorRKHS}, a key factor in choosing the interpolation points $(\vx_i)_{i=1,\dots,n}$ is the grid density measure $h$, which is defined by
\begin{align}\label{eq:griddensitymeasure}
    h:=\sup\limits_{x\in \Omega_n}\min\limits_{1\leq j\leq M}\norm[\vx-\vx_j]_2,
\end{align}
and which gives the largest distance between a point $x\in\Omega_n$ and the closest grid point, cf. \cite{Wendland98}. The set $\Omega_n$ here denotes the largest connected domain containing $\cup_{i} \mathrm{supp}(\phi(\cdot - \vx_i))$ and such that $\mathrm{supp} (f) \subset \Omega_n \subseteq \Omega$. 
The grid points $(\vx_i)_{i=1,\dots,n}$ can form regular shapes resulting, e.g., in rectangular, circular or elliptic grids as illustrated in \Cref{fig:grid_types}, but also other types of point clouds could be possible. Nevertheless, we assume that the grid is \textit{regular enough}, \textit{i.e.}  $\min_{i \neq j}\Vert \vx_i -\vx_j\Vert \lesssim \max_{i\neq j}\Vert \vx_i -\vx_j\Vert$, so that $\Psi^\gamma$ can be stably computed. This consideration is difficult to quantify, but choosing extremely irregular scattered points can lead to numerical instabilities since the domain $\Omega$ is then not uniformly encoded in the forward model, which may hinder the accurate representation of spatial structures, see also stability considerations as discussed in \cite{wendland05,chen2005scientific}.

\begin{figure}[h]
    \begin{subfigure}{0.3\textwidth}\centering
        \includegraphics[width=0.8\textwidth]{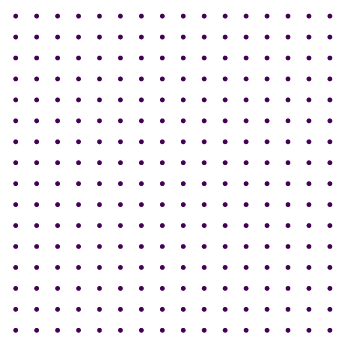}      
        \caption{}
        \label{fig:sq_grid}
    \end{subfigure}\hspace{0.04\textwidth}
    \begin{subfigure}{0.3\textwidth}\centering
        \includegraphics[width=0.8\textwidth]{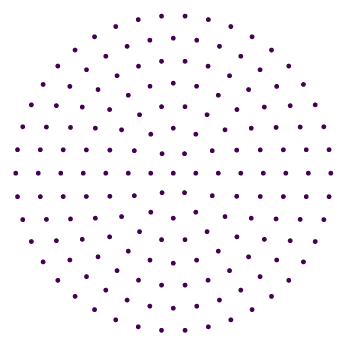}   
        \caption{}       \label{fig:circgrid}        
    \end{subfigure}\hspace{0.04\textwidth}
    \begin{subfigure}{0.3\textwidth}\centering
        \includegraphics[width=0.6\textwidth]{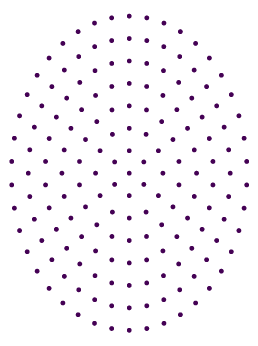}       
        \caption{}  \label{fig:ellgrid}         
    \end{subfigure}  
        \caption{Example of interpolation grids}    \label{fig:grid_types}
\end{figure}

				    \par
 
\subsubsection*{Prerequisites on the basis function $\phi$} 
We provide a short overview regarding the choice of $\phi$ and resulting interpolation properties of $\Pi_n$. For more details, we refer, e.g., to \cite{wendland05}.

To exploit the invariance properties of the Radon transform, we consider $\phi$ {\color{black} $\in C(\mathbb{R}^2) \cap L_1(\mathbb{R}^2)$} to be a \emph{radial} basis function (hence isotropic and rotation invariant). We also allow the basis function to depend on some parameter, i.e. we write in the following $\phi(\,\cdot\,;\mu)$ with $\mu>0$. In the context of reproducing kernel Hilbert spaces, $\mu$ is typically the scaling parameter.

For $\mu$ fixed, $\phi(\,\cdot\,;\mu)$ shall further be \emph{positive definite}, i.e. for any set of distinct points $\lbrace \widetilde{\vx}_j \ : \ 1\leq j \leq n \rbrace$, the quadratic form given by
\begin{equation}
	\mathcal{Q}\colon \CC^n\times \CC^n\,,\, (\text{u},\text{v})\mapsto \sum\limits_{k,l=1}^n u_k\bar{v_l} \, \phi\left(\vx_l-\vx_k; \mu\right)\label{eq:spd_cond}
\end{equation}
is positive definite. 
On the one hand, this property ensures that the system
\begin{equation}
	\label{eq:coeff_sys} f(\vx_j)=\Pi_n \fvec (\vx_j), \quad j=1,\dots,n
\end{equation}is uniquely solvable (for $f$ that can be evaluated pointwise). On the other hand, this further implies that the Fourier transform of $\phi(\,\cdot\,;\mu)$ is real-valued and positive as shown in \cite{wendland05}, Theorem 6.11.

Error bounds of type \eqref{eq:interpolerrorRKHS} can be obtained for $f$ belonging to the so-called \emph{native space}
\begin{equation}
	\mathcal{N}_\phi(\RR^2):=\left\{ f \in L_2(\Omega)\cap C(\RR^2) \, : \,  \int\limits_{\RR^2} \left|\widehat{f}(\xi)\right|^2\left(\widehat{\phi}(\xi;\mu)\right)^{-1}\diff\xi<\infty\right\} \label{eq:native_space}
\end{equation}
	Equipped with the inner product
	\begin{align*}
				\bil{f}{g}_{\mathcal{N}_\phi(\RR^2)}:=\frac{1}{2\pi} \int\limits_{\RR^2} \hat{f}(\xi)\overline{\hat{g}(\xi)}\left(\hat{\phi}(\xi;\mu)\right)^{-1},
	\end{align*}
	the native space becomes a Hilbert space, see \cite{wendland05}, also referred to as \emph{reproducing kernel Hilbert space}. We denote the associated norm by $\norm{}_{\phi}$. 
    Elements of $\mathcal{N}_\phi(\Omega)$ are called \emph{dominated by $\phi$}. 
										
    \begin{remark}
    Please note that the author in \cite{wendland05} constructed the native space differently and then showed     the equivalence to $\mathcal{N}_\phi$ as given above, Further note that, for $\Omega\subseteq\RR^2$, one can identify functions of the corresponding native space $\mathcal{N}_\phi(\Omega)$ with functions of $\mathcal{N}_\phi(\RR^2)$ using natural restrictions and extensions, c.f. \cite{wendland05}. \end{remark}

In the following, we consider $\phi$ to satisfy in addition to being positive definite the following property which is essential for our purposes. 

\begin{assumption}\label{assumption:phi}
$\phi(\,\cdot\,;\mu)$ shall admit the existence of constants $C, h_0>0$ such that for every set of interpolation points $(\vx_i)_{i=1,\dots,n}$ with grid density measure $h\leq h_0$ and for any 
		$f$ dominated by $\phi$, it holds  	 	\begin{equation}\label{eq:upper_bound}
	 		\| \Pi_n \fvec(\vx)-f(\vx)\|_{\infty}  \leq C\, P_\phi(h) \, \|f\|_\phi 
	 	\end{equation}
        with function $P_\phi$ and a constant $C$ independent of $f$ and $h$.
\end{assumption}
We provide some examples for basis functions that satisfy this property.

\begin{itemize}
    \item Assumption \ref{assumption:phi} is satisfied by positive definite radial basis functions which are sufficiently smooth, more precisely, $\phi(\, \cdot \, ;\mu) \in C^k$  and its derivatives of order $k$ satisfies $D^\alpha_\vx \phi(\vx;\mu) = \mathcal{O} (\norm[\vx]_2^\nu)$ for $\norm[\vx] \to 0$ with a $\nu > 0$. In this case, it holds $P_\phi(h)=h^{(k+\nu)/2}$, see \cite{wendland05}, Theorem 11.11.
    \item The Gaussian function satisfies assumption \ref{assumption:phi} with $P_\phi(h)=\exp{(-c/h)}, \, c>0$, i.e. the decay of the interpolation error is exponential in $h$, see \cite{wendland05}, Theorem 11.22. 
    \item 
    Despite its fast convergence property, the Gaussian suffers the drawback of not having a compact support. This can result in stability issues regarding the actual interpolation (depending on the chosen grid), see \cite{wendland05}, and will result in a dense projection matrix $A_\Phi$. Therefore, basis functions with compact support would be favorable. This further guarantees that the interpolant $\Pi_n \fvec$ is compactly supported as well, which will facilitate the computation of our reconstruction error. For this reason, we propose to make use of the \emph{Wendland functions} $\phi^{2,k}$ introduced in \cite{Wendland95}, which are compactly supported and positive definite as demonstrated in \cite{Wendland95}.
 
    Starting with the truncated power function $\phi^l(r)=\max(1-r,0)^l$ one can define
    \begin{align} \label{eq:dim_walk_op}
		\phi^{2,k}=I^k\phi^{2+k} \quad \text{with} \quad  (I \phi)(r)=\int\limits_{r}^\infty t\phi(t)\d t,
    \end{align}
 	 		    which gives a sequence of positive definite functions, see \cite{Wendland95}.
	We can now consider the radial basis function 
    $$
    \phi(\vx;\mu,k):=    \phi^{2,k}\left(\frac{\norm[\vx]_2}{\mu}\right)
        $$     with scale parameter $\mu>0$ and $k\geq 1$, which serves as additional smoothness parameter (the larger $k$, the faster the decay of the Fourier transform of $\phi(\,\cdot\,;\mu,k)$). 
    
    These basis functions satisfy assumption \ref{assumption:phi} for $f$ sufficiently smooth with $P_\phi(h)=h^{2k+1}$, see \cite{wendland05}, Theorem 11.23.
             		  	 						
														    \begin{remark}
    In \cite{Wendland98}, it was shown that 
				        the native space \eqref{eq:native_space} of $\phi(\cdot;\mu,k)$ corresponds to the Sobolev space		
					$H^{k+\frac32}\left(\RR^2\right)$
				and that the norms $\norm_\phi$ and $\norm_{H^{k+\frac32}}$ are equivalent. 										        This delivers a precise characterization of elements in $L_2$ that are dominated by the Wendland functions. In particular, this means that for sufficiently large $k$ any dominated function has a continuous representative meaning pointwise evaluations are indeed possible. 
        \end{remark}
 
\begin{remark}
    In order to compute our discrete operator $A_\Phi$, we have to evaluate the Radon transform of the Wendland functions. These can be stated explicitly as worked out in the Appendix \ref{app:R_monomial}.    
\end{remark}
\end{itemize}
 
\subsubsection{Reconstruction error for dominated functions}
 
With the general approximation property of $\Pi_n$, we directly obtain
the following error estimate for the semi-discrete data associated to $f$ and $\Pi_n \fvec$.

\begin{corollary}
    Let $\phi$ be a compactly supported, positive definite radial basis function satisfying assumption \ref{assumption:phi}. Further, let $f$ be a function dominated by $\phi$ with $\Pi_n \fvec$ denoting the approximation of $f$ in the basis $(\phi(\cdot-\vx_i))_{i=1,\dots,n}$ with sufficiently small grid density measure $h$.\vspace{1ex}\\
    Then, it holds
    \begin{equation}
		\norm[\Xi_m \R_\Phi (\Pi_n\fvec-f)]_{\RR^m}\leq  2 \sqrt{m}\,  C \norm[f]_\phi  P_\phi(h) .\label{eq:error_bound}
	\end{equation}
\end{corollary}

\begin{proof}
    Let $\mathbb{1}_{\Omega}$ denote the characteristic function of the disk. Then, for $l=1,\dots,m$, it holds with     assumption \ref{assumption:phi} and $\textup{diam}(\Omega)\leq 2$
    \begin{align*}
        |\R_\Phi (\Pi_n\fvec-f)((s,\theta)_l)| \leq \norm[\Pi_n\fvec-f]_\infty \, |\R_\Phi \mathbb{1}_{\Omega}((s,\theta)_l)| \leq 2\,  C\, \norm[f]_\phi \, P_\phi(h).
    \end{align*}
    Applying the Euclidean norm on the vector
    \begin{align*}
        \left(\R_\Phi (\Pi_n\fvec-f)((s,\theta)_1),...,\R_\Phi (\Pi_n\fvec-f)((s,\theta)_m)\right)^\top
    \end{align*}
    yields the result.    
\end{proof}
 
Now, we can prove the following upper bound for the overall reconstruction error. 

\begin{theorem} \label{theorem:recerror}
Let $f$ be dominated by a compactly supported, positive definite radial basis function $\phi$ satisfying assumption \ref{assumption:phi} and denote $\gvec = \Xi_m \R_\Phi f$. 
Let further $E^\gamma$ be the discrete mollifier defined by \eqref{eq:semi_disc-moll} and let $\Psi_h^\gamma$ denote the solution to $A_{\Phi,\Xi_m,\Pi_n}^\top \Psi^\gamma = (E^\gamma)^\top$ with $\Pi_n$ as in \eqref{eq:interpol_op} with sufficiently small grid density measure $h$.\\
Then, the overall reconstruction error in the reconstruction points satisfies
		\begin{align}
			\left\|\left(\Psi_h^\gamma\right)^\top \gvec-\fvec^\gamma_\vz\right\|_{\RR^r}
            \leq C\norm[f]_{\phi} 
            \left(2\sqrt{m} \, \norm[\Psi_h^\gamma]_2+\sqrt{r}\right)  P_\phi(h)
            \label{eq:csrbf_rec_err}
		\end{align}
		with constant $C$ independent of $f,\, h, \gamma,m,r$.
  	\end{theorem}
 
\begin{remark}
    Please note that the solution of the auxiliary problem with projection matrix $A_{\Phi,\Xi_m,\Pi_n}$ depends on $m$ and $n$, i.e. in particular on the number of interpolation points and hence on the grid density measure $h$. Since the dependence on $h$ is crucial when interpreting the upper bound for the reconstruction error, we indicate this explicitly by adding an index $h$ to the reconstruction kernel, i.e. by writing $\Psi^\gamma_h$. \end{remark}

	\begin{proof} Let $\Pi_n \fvec$ denote the approximation to $f$ by $\phi$ and a grid with density measure $h$. Further, we define $\gvec_n := \Xi_m \R_\Phi \Pi_n \fvec = A_\Phi \fvec$. It holds
\begin{align*}
			\left\|\left(\Psi_h^\gamma\right)^\top \gvec - \fvec^\gamma_\vz\right\|_{\RR^r}&\leq 
			\left\|(\Psi_h^\gamma)^\top (\gvec - \gvec_n)\right\|_{\RR^r}+\left\|(\Psi_h^\gamma)^\top(\gvec_n)-\fvec^\gamma_\vz\right\|_{\RR^r}.
		\end{align*}	
This leads to two terms controlling the reconstruction error: 
\begin{itemize}
    \item For the first term, which quantifies the  interpolation error propagated by $\Psi_h^\gamma$ and $\R_\Phi$, we obtain $$\left\|(\Psi_h^\gamma)^\top (\gvec - \gvec_n)\right\|_{\RR^r} \leq \norm[\Psi_h^\gamma] \norm[\Xi_m \R_\Phi(\Pi_n\fvec-f)]_{\RR^m}\leq 2\sqrt{m}\,C\,\norm[\Psi_h^\gamma] \ \norm[f]_{\phi}\, P_\phi(h).$$ 
    \item Now we consider the second term.     
        By construction we have $$(\Psi_h^\gamma)^\top A_\Phi \fvec_n = E^\gamma \fvec,$$ i.e. $\left\|(\Psi_h^\gamma)^\top(\gvec_n)-\fvec^\gamma_\vz\right\|_{\RR^r} = \left\| E^\gamma \fvec-\fvec^\gamma_\vz\right\|_{\RR^r}$. With $$(E^\gamma \fvec)_l - (\fvec_{\vz}^\gamma)_l = \langle \Pi_n \fvec - f, e^\gamma_{\vz_l}\rangle,$$
        cf. Section 3.1, this term describes the approximation error propagated by the mollifier. 
        With assumption \ref{assumption:phi}         and $\| e^\gamma_{\vz_l}\|_{L_1} = 1$, we obtain for each $l=1,\dots, r$
        $$ |(E^\gamma \fvec)_l - (\fvec_{\vz}^\gamma)_l| = |\langle \Pi_n \fvec - f, e^\gamma_{\vz_l}\rangle| \leq
         C \norm[f]_{\phi}  P_\phi(h),$$
         i.e. $\left\| E^\gamma \fvec-\fvec^\gamma_\vz\right\|_{\RR^r}\leq \sqrt{r} \, C \norm[f]_{\phi}  P_\phi(h)$.
    \end{itemize}
    Combining both estimates yields the asserted statement. 	\end{proof}  
 
This error estimate offers insights into how the different factors affect the actual reconstruction error. 
    While the interpolation error - specifically the interplay of $C$, $\norm[f]_\phi$ and $P_\phi(h)$ - has been thoroughly studied in the literature on approximation theory, in particular for the basis functions given as examples above, we focus here on the interaction of $P_\phi(h)$ with the additional factor $\left(2\sqrt{m} \, \norm[\Psi_h^\gamma]_2+\sqrt{r}\right)$.     Large values of $m$ and $r$, i.e. a large number of measurements and reconstruction points, can be compensated for by choosing a grid with sufficiently small grid density measure $h$. 
        However, the smaller $h$, the larger becomes $\norm[\Psi_h^\gamma]$. This is due to the ill-posedness of the underlying inverse problem (which implies that $A_{\Phi}$ is increasingly ill-conditioned). Although we do not have an explicit asymptotic estimate of $\norm[\Psi_h^\gamma]$ with respect to $h$,     we observe in our numerical experiments with the Wendland functions that the product $\norm[\Psi^\gamma_h] \, P_\phi(h)$ indeed stays bounded and still tends to zero as $h\to 0$, see Figure \ref{fig:Ploterrorestimate}, at least as long as $k$ is chosen large enough. In Figure \ref{fig:Ploterrorestimate}, we plotted $\norm[\Psi_h^\gamma]$ for an increasing number of grid points $n$ associated to the grid density measure $h$ for three different regular grid types (rectangular, circular and ellipse).     We clearly observe that the norm increases, with the rectangular grid showing the steepest increase. In addition we plotted $(P_\phi(h))^{-1}$ for the Wendland function $\phi^{2,k}$ with $k=5$, i.e. $P_\phi(h)=h^{2k+1}=h^{11}$.     
        The comparison shows that each $\norm[\Phi_h^\gamma]$ increases significantly slower than $(P_\phi(h))^{-1}$. Thus, $\norm[\Psi_h^\gamma] \, P_\phi(h)$, despite the increase of the first factor, overall still tends to zero as $h\to 0$. In particular, the same is expected for other (smoother) radial basis functions with even faster decaying $P_\phi$. Also note that no additional filter $F_\tau$ has yet been added in the computation of $\Psi_h^\gamma$, which will further attenuate the increase of $\norm[\Psi^\gamma_h]$. The comparison in Figure \ref{fig:Ploterrorestimate} also reveals that the increase of $\norm[\Psi_h^\gamma]$ can be softened by the grid choice. Although it suggests that the elliptic grid results in smaller condition numbers than circular or rectangular grid, the grid should in practice be chosen in accordance to the studied object. For instance, circular objects might be overall better represented by a circular grid, hence resulting in smaller approximation and ultimately smaller reconstruction errors.     
 
\begin{figure}[!h]
    \centering
    \includegraphics[width=0.9\linewidth]{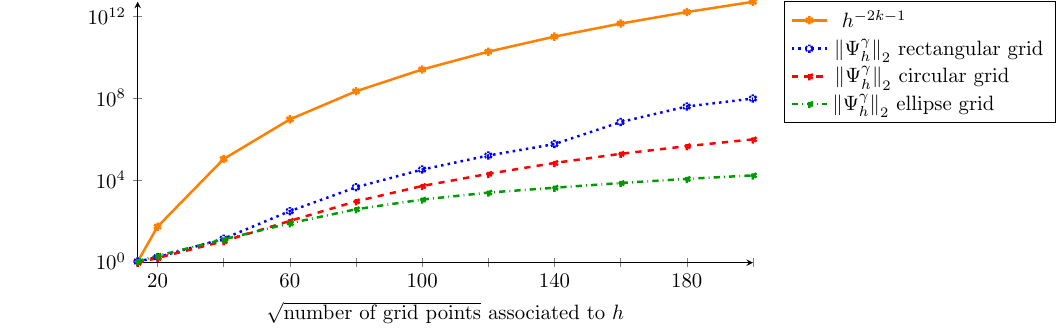}    \caption{$\norm[\Psi_h^\gamma]$ (dashed/dotted lines) for increasing number of interpolation points for regular grids, as well as $(P_\phi(h))^{-1}=h^{-2k-1}$ (solid line) for the associated grid density $h$ with $P_\phi$ for the Wendland function with $k=5$.  }
    \label{fig:Ploterrorestimate}
\end{figure}
 
\subsubsection{Smoothing data to enforce dominance}
The estimates for the interpolation error, and hence for the reconstruction error, hold only for dominated functions, cf.  assumption \ref{assumption:phi}. However, it is clear from the definition of $\mathcal{N}_\phi(\Omega)$, that not every $L_2$-function can be dominated by a given $\phi$. Nevertheless, if a given $f\in L_2(\Omega)$ fails to be dominated by $\phi$, this property still applies to its mollified counterparts, as the following lemma shows. 
\begin{lemma}
	\label{lem:dominate}
	Let $f\in L_2(\RR^2)$ be compactly supported in $\Omega$, let $\phi$ {\color{black}$\in C(\RR^2)\cap L_1(\RR^2)$ be a positive definite radial basis function}, and let $\chi\in C(\RR^2)\cap L_2(\RR^2)$ be a mollifier. 	If the Fourier transforms 	of $\phi$ and 	$\chi $ satisfy
	\begin{equation}
		{\left|\widehat{\chi}(\xi)\right|^2}{\,\left(\widehat{\phi}(\xi)\right)^{-1}}\leq c_{\phi,\chi}\text{ a.e.}\label{eq:smooth_cond}
	\end{equation}
	for some $c_{\phi,\chi}\in\RR^+$, then the convolution $f*\chi$ is dominated by $\phi$. \end{lemma}

\emph{Remark:} Please note that $\chi$ is independent of the mollifier considered for the reconstruction process.
\begin{proof}
	With the standard convolution theorem, we obtain 
	\begin{align*}\frac{1}{4\pi^2}\int\limits_{\RR^2}\left|\widehat{f\ast \chi}(\xi)\right|^2\,\left(\widehat{\phi}(\xi)\right)^{-1}\d\xi 
		&=\frac{1}{4\pi^2}\int\limits_{\RR^2}\left|2\pi\,\widehat{f}(\xi)\, \widehat{\chi}(\xi)\right|^2\cdot \left(\widehat{\phi}(\xi)\right)^{-1}\d\xi\leq c_{\phi,\chi}\norm[f]_{L_2}^2,
	\end{align*}
	hence $f\ast \chi$ is dominated by $\phi$. \end{proof}	

We want to discuss two examples for pairs of $\phi$ and $\chi$ leading to dominance:
\begin{itemize}
	\item If the Fourier transform $\widehat{\phi}$ of the basis function $\phi$ is continuous on $\RR^2$ (including $0$), then we can choose $\chi = \phi$, since 	$c_{\phi,\chi}:=\sup\limits_{\xi\in \RR^2}\left|\widehat{\phi}(\xi)\right|<\infty$
	due to the boundedness of $\widehat{\phi}$. Consequently, one can always project an $L_2$ function into the space of dominated functions (i.e. $L_2(\Omega)$ can be continuously embedded into $\mathcal{N}_\phi(\Omega)$.\\
						\item Consider $\phi$ and $\chi$ chosen as Gaussians, more precisely 		$$
	\ds \phi(\vx;\mu)=\frac{1}{2\pi \mu	}e^{-\tfrac{\norm[\vx]^2_2}{2\mu}} \text{ and } \ds \chi(\vx;\nu)=\frac{1}{2\pi \nu}e^{-\tfrac{ \norm[\vx]^2}{2\nu}}
	$$ 
	with parameters $\nu, \mu >0$. With their respective Fourier transforms $\ds\widehat{\phi}(\xi;\mu)=\frac{1}{2\pi} e^{-\tfrac{\mu^2}{2}\norm[\xi]^2}$ and $\ds \widehat{\chi}(\xi;\nu)=\frac{1}{2\pi} e^{-\tfrac{\nu^2}{2}\norm[\xi]^2}$, it holds  $$\ds 
	{\Big|\widehat{\chi}(\xi;\nu)\Big|^2}{\, \left|\widehat{\phi}(\xi;\mu)\right|^{-1}}=\frac{1}{2\pi} e^{\left(\tfrac{\mu}{2}-\nu\right)\norm[\xi]^2}.$$ Hence $f*\chi(\,\cdot\,;\nu)$ is dominated by $\phi(\,\cdot\,;\mu)$ whenever $\nu\geq \frac{\mu}{{2}}$.	Hence, the additional smoothing according to $\chi$ to ensure dominance needs to be stronger than the smoothing stemming from the basis function. 	\end{itemize}
 
The dominance of $f$ by $\phi(\cdot,\mu)$ is necessary to guarantee a controlled reconstruction error in \Cref{theorem:recerror}.According to \Cref{lem:dominate} one can consider the convolution by $\phi_\nu \ast  f$ -- where $\phi_\nu$ is chosen so that $\vert\hat{\phi_\nu}\vert^2(\hat{\phi}(\cdot;\mu))^{-1}$ is bounded -- in the forward model to enforce it. This corresponds to the following rework of the forward problem:
$$ 
\Xi_m  \R_\Phi f^\nu 
:= \Xi_m  \R_\Phi \left( \phi_\nu \ast  f \right)
= \Xi_m  \left( \bar\phi_\nu \ast  \R_\Phi f \right)
= \Xi_m (\bar\phi_\nu \underset{s}{\ast} g)
=: \Xi_m g^\nu
\approx \Xi_m (\bar\phi_\nu) \underset{m}{\ast} \gvec
=:\gvec^\nu
$$
where $\bar{\phi}_\nu(s) := \R_\Phi\phi_\nu(s,\mathbf{0})$ using that $\phi_\nu$ is radial and $\underset{m}{\ast}$ denotes the discrete convolution with respect to the $(s_j)_j$. We note that $\Xi_m g^\nu\neq \gvec^\nu$ here since $g$ is compactly supported and therefore not band-limited. However, due to the smoothness properties of $\R_\Phi$ -- which maps $L_2(\Omega)$ into $H^{\frac12}(Z_\Phi,w^{-1})$ -- and decay of the Fourier transform of $g$, one can expect a small approximation error in the discretization process.  
\\

With this general setup, we can now consider noisy data $\gvec_\delta = \gvec + \delta$ instead of $\gvec$. Obviously, one can write
\begin{align*}
\Vert (\Psi^\gamma)^\top \gvec_\delta^\nu - \fvec^{\nu,\gamma}_{\vz} \Vert_{\RR^r}
&\leq
\Vert (\Psi^\gamma)^\top (\gvec_\delta^\nu - \Xi_m g^\nu) \Vert_{\RR^r} +
\Vert (\Psi^\gamma)^\top \Xi_m g^\nu - \fvec^{\nu,\gamma}_{\vz} \Vert_{\RR^r} \\
&\leq
\Vert (\Psi^\gamma)^\top \Vert_{\RR^m \to \RR^r} \left( \Vert \Xi_m (\bar\phi_\nu) \underset{m}{\ast} \delta \Vert_{\RR^m} +  \Vert \gvec^\nu - \Xi_m g^\nu \Vert_{\RR^m}\right) +
\Vert (\Psi^\gamma)^\top \Xi_m g^\nu - \fvec^{\nu,\gamma}_{\vz} \Vert_{\RR^r}.  
\end{align*}
The second term on the right-hand side is treated by \Cref{theorem:recerror}. The first term estimates the impact of the LARK on the noise level after smoothing and on the aliasing error in the discretization process $\Xi_m g^\nu$.  The {\color{black} is a classic problem in numerical analysis and is controlled by the smoothness of $\phi_\nu$ and choosing $\nu$ appropriately.
 
Regarding the noise level $\delta$, it is clear with such an ill-conditioned problem that it must be as small as possible and will force a trade-off between stability and resolution here parametrized by $\nu$.
 
\begin{remark} We want to quickly discuss how this pre-smoothing strategy to ensure dominance relates to CLARK. The latter, as introduced in \Cref{sec:denoising}, supports the LARK by constraining the solution space and corresponds to a denoising (hence a pre-smoothing) step of the data $\gvec^\delta$. In the continuous case, when using the total-variation penalty functional in \Cref{eq:reg-AI-discrete}, we map $\S_\gamma(g)$ into the space of bounded variation (BV) functions. Assuming that $\S_\gamma(\D_\lambda(g^\delta))$ is additionally absolutely continuous, then its Fourier transform is $o(\Vert \xi\Vert^{-1})$. Since $\S_\gamma$ is continuous, this property on the decay of the Fourier coefficients translates into the denoised data $\D_\lambda(g^\delta)$. This decay might not be enough to ensure dominance in general (cf. Definition 3). However, as we will see in the numerical results below, it appears that when working with real data and choosing $\phi$ as a Wendland function, the \textit{smoothing} induced by the total-variation term in the CLARK suffices.
\end{remark}
 
\subsubsection{Numerical results}
 
First, we want to illustrate that the described interpolation approach with sufficiently smooth radial basis functions, like Gaussian or Wendland function, solves the representation issue described in the beginning of the Section. To this end, Figure \ref{fig:LARK_on_diff_grids} shows LARK and CLARK reconstructions for $30^\circ$ and $40^\circ$ missing angles for the cylindrical real data from IzfP and for the analytic Shepp-Logan data for different regular grids, namely rectangular, circular and ellipse grid. In comparison with the results in Fig. \ref{fig:AluSi_representation_issue} regarding the pixel grid interpolation, we see a substantial improvement. The missing features are better recovered and the limited-angle artefacts significantly reduced. When using LARK, there remain the wave-type artefacts stemming from the singular functions to small singular values which however can be further compensated for by the CLARK approach. The reconstruction results, in particular those obtained with LARK alone, also offer insight into the influence of the grid choice. LARK removes the limited-angle artefacts and missing regions slightly better for a circular grid in the case of cylinders \Cref{fig:LARK_on_diff_grids}(b)) and for an elliptic grid in the case of the Shepp-Logan (\Cref{fig:LARK_on_diff_grids}(f)). We add that we choose the same major and minor axes as for the largest ellipse in the Shepp-Logan phantom but did not enforce the grid points to match the edges. {\color{black}The same applies to the circular grid for the circular object regarding the Fraunhofer data.} We therefore have here an intrinsic prior on shape but not position. The right grid for the right shape is expected to help the accuracy of the operator $\Pi_n$ and deliver better representability of the data. It is important to note that for complex objects such as the Shepp-Logan or later tests, the grid tends to be less important for the inside features. This can be explained by the importance of the edges of $f$ in the data $\gvec$ which are encoded in strong discontinuities. 

\begin{figure}[h!]
    \centering
    \includegraphics[width=\linewidth]{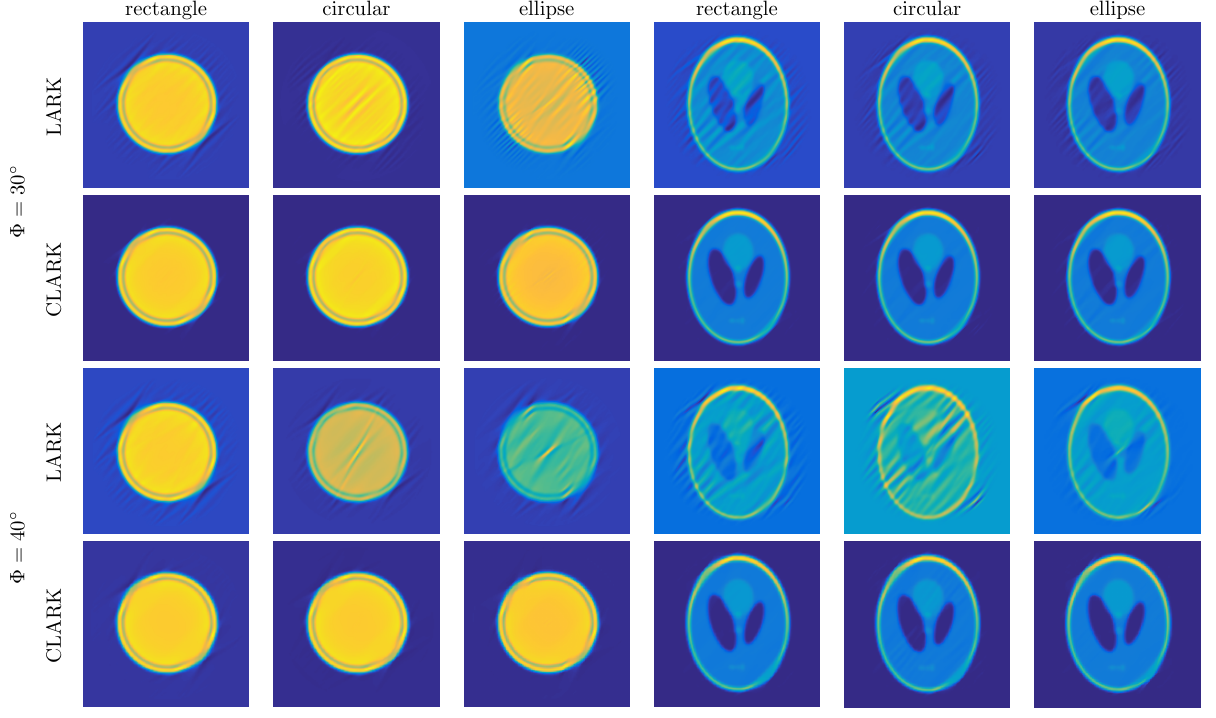}
    \caption{LARK and CLARK reconstructions for $\Phi=\{30^\circ,40^\circ\}$ for the real data from IzfP and for the analytic data on the Shepp-Logan phantom for different grids.}
    \label{fig:LARK_on_diff_grids}
\end{figure}
 
To illustrate the potential of our approach, we further provide reconstruction results from real data for a more complex object than the symmetric cylinder, more precisely from a real data set measured for the Helsinki tomography challenge 2022 
\cite{helsinkiDataSet,htc22}. Figure \ref{fig:csrbf_helsinki_compAngularComp} shows the respective results of standard FBP as well as of both LARK and CLARK for different $\Phi$ ranging from $10^\circ$ to $50^\circ$ missing angular range. In light of the geometry of the object, we chose a regular circular grid for the interpolation. Despite the data imperfections (measurement noise, model imperfections, etc.), the missing features are in fact recovered and the characteristic limited-angle streak artefacts suppressed. As observed before, CLARK further reduces the wave-type artefacts present in the LARK-reconstructions. Only for $\Phi=50^\circ$, some features in the interior begin to fade and the density values can no longer be recovered correctly. Nevertheless, in comparison with the classical FBP result, LARK and CLARK still recover most missing edges for $\Phi=50^\circ$.  Overall, the results demonstrate that our approach can in fact handle limited data from real measurements. The precise angular range, up to which good results can be achieved, depends on the noise level and potential model imperfections, reflecting the severely ill-posed nature of the underlying inverse problem.
 
\begin{figure}
	\includegraphics[width=0.9\textwidth]{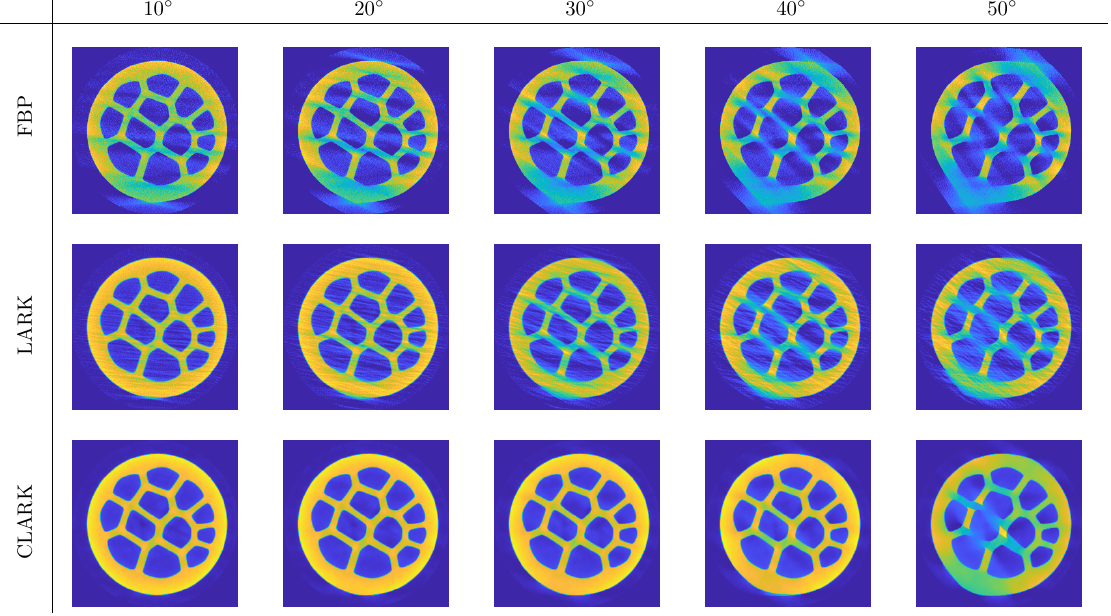}
	\caption{\label{fig:csrbf_helsinki_compAngularComp} Reconstruction results for real data from the Helsinki Tomography Challenge 2022 for different angular ranges.}
\end{figure}
 
\section{Conclusion}

This article provides a solver for limited-angle CT based on the method of the approximate inverse -- a versatile framework that designs a regularized inverse by using reconstruction kernels obtained as the solution of suited auxiliary problems. The kernels for limited-angle CT (LARK) were computed using the singular value decomposition of the forward operator. The inherent stability issues due to the severely ill-posed nature of the problem were tackled fourfold to enable applicability to real measured data: 
 
\begin{enumerate}[label=(\roman*)]
    \item the smallest singular values are attenuated via a spectral filter,
    \item the prescribed mollifier controls the high frequency components of the solution,     \item a pre-smoothing step on the measured data enforces representability of the function $f$ in a basis of smooth radial basis functions,
    \item finally, an additional denoising step on the data utilizing, for instance, the total-variation functional, further constrains the solution space, resulting in the constrained kernel approach (CLARK). \end{enumerate}
Our numerical validation on synthetic and real data demonstrates that this whole construction can solve the limited-angle problem for substantial angular restrictions but will show limits for large restrictions with 
significant noise levels. We can mention that the total-variation penalty terms used here in the CLARK remains limited, as it cannot distinguish between edges and wave artefacts. At the core of our future research, the design of a more suitable \textit{range projector} and denoising functionals could improve the stability of the approach and allow improved reconstructions also for larger $\Phi$. More generally, 
supplementary information is required to solve the problem further, as is done with deep learning techniques. With this in mind, the proposed method can serve as a new foundation for the further development of modern, learning-based algorithms. 

\section*{Acknowledgements}

The authors would like to thank and remember late Prof. Dr. Alfred K. Louis who pioneered this work in many ways, from deriving the singular value decomposition of the limited-angle Radon transform in 1986, building the Approximate Inverse ten years later to  {\color{black} engaging in valuable discussions in the early stages of this research work.} 

The work of the first and second author was funded by Deutsche Forschungsgemeinschaft (DFG, German Research Foundation) under Germany’s Excellence Strategy—EXC 2075 - 390740016. The first author was further supported by the Sino-German Mobility Programme (M-0187) by the Sino-German Center for Research Promotion.\\
 
\printbibliography

\appendix
\section{Appendix}
 
\subsection{Radon Transform of radial monomials} \label{app:R_monomial}

In order to compute our discrete operator $A_\Phi$ we have to evaluate the Radon transform of the Wendland functions. As can be seen in \cite{wendland05}, Theorem 9.12, the functions $\phi(\cdot;\mu,k)$ can be represented by the radial polynomial 
        \begin{align*}
               \phi(\vx;\mu,k) &= c_\mu \sum\limits_{l=0}^{3k+2}d_{l,k} \left(\frac{\Vert \vx \Vert}{\mu}\right)^l \quad \text{for }\Vert \vx \Vert \leq \mu,
    \end{align*}
    where the coefficients $d_{l,m}$  satisfy     the recurrence $d_{l,m+1}= -\dfrac{d_{l-2,m}}{l}$ for $0\leq m\leq k-1$  and $2\leq l\leq k+2m+4$ with
$$
d_{l,0}=(-1)^l \binom{k+2}{l},\quad 0\leq l\leq k+2, \qquad
d_{0,m+1}= \sum\limits_{l=0}^{k+2m+2}\dfrac{d_{l,m}}{l+2},  
\qquad \text{and} \quad d_{1,m+1}=0.
$$     
                            Hence it is sufficient to compute the Radon transform of radial monomials.

\begin{lemma}\label{lem:R_monomial}
Let $f^{(l)}(\vx) = \Vert \vx \Vert ^{l} \chi_\alpha(\vx)$, with $\chi_\alpha$ the characteristic function of the disk $\{\vx \, :\, \Vert\vx\Vert\leq \alpha\}$. Then the Radon transform of $f^{(l)}$,  for $l \geq -1$ and $\vert s\vert < \alpha$, satisfies
$$
\R[f^{(l)}](s) = 2 \vert s \vert^{l+1} \int_0^{\arccos\frac{\vert s \vert}{\alpha} } \sec^{l+2} \theta \, d\theta = 2 \vert s \vert^{l+1} I_{l+2}\left(s,\alpha\right)
$$
where $I_k(s,\alpha)$ satisfies for \( k \ge  1 \) the recurrence  
$$
I_{k+2}\left(s,\alpha\right)
= \frac{1}{k+1} \left( \left(\frac{\alpha}{\vert s\vert}\right)^{k} \sqrt{\frac{\alpha^2}{s^2} - 1} + k\cdot  I_{k}\left(s,\alpha\right) \right)
$$
with $I_1(s,\alpha) = \cosh^{-1}\left(\alpha/s\right)$ and $I_2(s,\alpha) = \sqrt{\frac{\alpha^2}{s^2} - 1}$.

\end{lemma}

\begin{proof} 
For radial functions, the Radon transform simplifies to
$$
\R[f^{(l)}](s) = 2 \int_{\vert s\vert}^\alpha \frac{r^{l+1}}{\sqrt{r^2-s^2}} \mathrm{d}r, \quad \vert s \vert  < \alpha.
$$
By integrating this expression via the substitution
$
r = \vert s \vert \sec \theta, \quad \theta \in \left[ 0, \arccos\left( \frac{\vert s \vert}{\alpha} \right) \right],
$
such that $dr = \vert s \vert \sec \theta \tan \theta \, d\theta$ and $\sqrt{r^2 - s^2} = \vert s \vert\tan \theta$, the integral yields
$$
\R[f^{(l)}](s) =2 \vert s \vert^{l+1} \int_0^{\arccos\frac{\vert s \vert}{\alpha} } \sec^{l+2} \theta \, d\theta = 2 \vert s \vert^{l+1} I_{l+2}\left(s,\alpha\right).
$$
Direct computations deliver
$$
I_1(s,\alpha) = \cosh^{-1}\left(\alpha/s\right) \qquad \text{and} \qquad
I_2(s,\alpha) = \sqrt{\frac{\alpha^2}{s^2} - 1}.
$$
The result follows from the standard recurrence for powers of secant, namely
$$
\int \sec^{l+2} \theta \, d\theta = \frac{\sec^{l} \theta \tan \theta}{l+1} + \frac{l}{l+1} \int \sec^{l} \theta \, d\theta ,
$$
which leads to the asserted recurrence
\begin{align*}
I_{l+2}\left(s,\alpha\right)
&= \frac{1}{l+1} \left( \left. \sec^{l} \theta \tan \theta \right\vert_{ \theta = \arccos\frac{\vert s \vert}{\alpha}} + l\cdot  I_{l}\left(s,\alpha\right) \right) 
= \frac{1}{l+1} \left( \left(\frac{\alpha}{\vert s\vert}\right)^{l+1} \sqrt{1 - \frac{\vert s \vert^2}{\alpha^2} } + l\cdot  I_{l}\left(s,\alpha\right) \right). \\
\end{align*}
\end{proof}

\end{document}